\documentclass[10pt,letterpaper]{article}

\usepackage{./kashsty}
\usepackage{./kashdef}
\usepackage{./kashthm}
\usepackage[round,authoryear]{natbib}
\usepackage{xcolor}
\usepackage[margin=1.37in]{geometry}
\usepackage{ulem}

\usepackage{algorithm}
\usepackage{algorithmic}

\usepackage{gensymb}

\def\PICDIR{./}

\title{Topological Hidden Markov Models}
\author{
  Adam B Kashlak, Prachi Loliencar, Giseon Heo\\
  \normalsize 
  Department of Mathematical and Statistical 
  Sciences\\
  \normalsize
  University of Alberta, Edmonton, AB
}

\begin{document}

\maketitle

\begin{abstract}
    The hidden Markov model (HMM) is a classic modeling tool
    with a wide swath of applications.  Its 
    inception considered observations restricted to a 
    finite alphabet, but it was 
    quickly extended to multivariate continuous 
    distributions.  In this article, we further extend
    the HMM from mixtures of normal distributions in 
    $d$-dimensional Euclidean space to general Gaussian
    measure mixtures in locally convex topological spaces.
    The main innovation is the use of the Onsager-Machlup
    functional as a proxy for the probability density
    function in infinite dimensional spaces.  
    This allows for choice of a Cameron-Martin space suitable
    for a given application.
    We demonstrate
    the versatility of this methodology by applying it
    to simulated diffusion processes 
    such as Brownian and fractional Brownian sample paths 
    as well as the Ornstein-Uhlenbeck process.
    Our methodology is applied to the identification of 
    sleep states from overnight polysomnography time
    series data with the aim of diagnosing Obstructive 
    Sleep Apnea in pediatric patients.  It is also 
    applied to a series of annual cumulative snowfall curves
    from 1940 to 1990 in the city of Edmonton, Alberta.
\end{abstract}

\tableofcontents

\section{Introduction}

The Hidden Markov Model (HMM) was and still is 
a powerful tool for modelling diverse datasets. 
Its inception dates to the work of Leonard 
Baum and his colleagues at the Institute for
Defense Analysis 
\citep{BAUM1966,BAUM1970} predating the 
advent of the expectation-maximization algorithm
\citep{DEMPSTER1977,wu1983convergence}.
Originally developed for observations taking 
values within a finite alphabet,
its applications included speech 
recognition and genome sequences
\citep{FERGUSON1980,PORITZ1986,JUANG1985,JUANG1991}. 
Since then, it has been extended to many types of data
including multivariate elliptically symmetric 
distributions \citep{LIPORACE1982} and  skewed 
distributions \citep{CHATZIS2010}.

When considering an ordered sequence of data
points, the HMM assumes that the sequence
of observed data is independent conditional
on a discrete state variable driven by a 
Markov chain.  For example, the daily 
number of people taking public transit 
to work (observation) may depend on the
weather (state).  If we consider two weather
states, snowy and sunny, we can model the 
day-to-day changes in the weather as a 
Markov chain with a $2\times2$ transition
matrix.  In this example, the weather is 
obviously observable.  However, the true 
power of the HMM is to learn hidden states
that may not be immediately obvious.  In this
way, the HMM aims to cluster observations,
but assuming only conditional independence
(as opposed to full independence) based on 
the states of a Markov chain.  An excellent
tutorial on HMMs can be found in 
\cite{RABINER1989} with many references to 
early work on this model therein.

The following work considers extending the observation
space of the Hidden Markov Model from Euclidean 
space to general locally convex topological 
vector spaces (LCTVS) specifically where the 
states of the Markov chain correspond to 
mean-shifted Gaussian measures.  As a result of
this general formulation, our so-called 
Topological Hidden Markov Model (THMM) can be 
adapted to many settings of interest from 
finite dimensional data to functional data 
and stochastic process data.  
In Section~\ref{sec:data}, we show the method's
applicability to a variety of simulated sample paths
from stochastic processes such as fractional Brownian 
motion and the Ornstein-Uhlenbeck process.
In Section~\ref{sec:pedOSA}, we apply these algorithms
to the task of identifying sleep states from 
an electroencephalogram (EEG) time series collected 
during overnight polysomnography.  
While this data is considered as a proof-of-concept
for the methodology,
the ultimate goal for subsequent research
is to be able to quickly identify sleep disorders
such as Pediatric Obstructive Sleep Apnea.
Section~\ref{sec:snowfall} applies both the classic HMM
and the new THMM algorithms to cumulative snowfall
curves from the city of Edmonton, Alberta with an 
aim of identifying patterns in snowfall across a 
50 year timespan.

In infinite dimensional spaces, there is no 
analogue of 
Lebesgue measure and thus no probability densities
and likelihood function to maximize.
This is the central challenge when working with 
functional and stochastic process data.
Our main innovation in this work is use the 
Onsager-Machlup functional for a general Gaussian 
measure \citep{BOGACHEV1998} within the
HMM emission function.  This allows for fitting
HMMs to data living in a wide variety of spaces such as 
$L^2[0,1]$ and the Sobolev space $W_0^{2,1}[0,1]$ among 
others.  There have been many past works on deriving the
Onsager-Machlup functional for various Gaussian processes
that we will make use of
\citep{TAKAHASHI1981,ZEITOUNI1989,SHEPP1992,CAPITAINE1995,IKEDA2014}.

There have been a few recent ventures taking the 
classic HMM into the realm of functional data.
In \cite{MARTINO2020}, they focus on multivariate 
functional data and choose the emission function to 
be the inverse squared $L^2$ distance.
The work of \cite{SIDROW2021} proposed a more complicated
hierarchical HMM model that can be used to partition 
high frequency functional data with complicated dependence 
structures so that the pieces can in turn be modelled 
via time series or functional data methods.
For analyzing longitudinal data, \cite{ALTMAN2007} proposes
a mixed HMM that incorporates covariates and random effects 
into the model.

The following Sections~\ref{sec:notation} and~\ref{sec:standardHMM}
introduce notation for locally convex spaces, semi-norms, 
Cameron-Martin spaces, and the classic Hidden Markov Model.
Section~\ref{sec:onsager} defines the Onsager-Machlup functional.
The two main algorithms, Baum-Welch and Viterbi, are outlined
in Section~\ref{sec:algorithm}.
Section~\ref{sec:interestingSpaces} details many 
specific models of interest including parametric
models like Brownian motion with linear drift and non-parametric
curve fitting.
Theorems justifying the validity of the THMM are 
stated and proved in Section~\ref{sec:theory}.  We
first prove that each step of the algorithm does, 
in fact, improve the analogue of the likelihood function
which is based on the Onsager-Machlup
functional.  We secondly prove that the sequence of 
reestimated parameters produced by the Baum-Welch algorithm
has at least one limit point and that all limit points 
of the sequence are critical points of the likelihood function.
Lastly, we show that in the dual sense, the corresponding
sequence of mixtures of Gaussian measures has a weak limit point
as the number of iterations of the algorithm tends to 
infinity.
Section~\ref{sec:data} demonstrates the power of the 
THMM applied to a variety of simulated datasets.
Section~\ref{sec:pedOSA} further demonstrates the power
of the THMM by showing its efficacy in identifying 
sleep states from noisy pediatric EEG data streams.
Section~\ref{sec:snowfall} considers climatological
data.
Future extensions to this work are briefly 
discussed in Section~\ref{sec:discussion}.

\subsection{Definitions and Notation}
\label{sec:notation}

A locally convex topological vector space (LCTVS)
generalizes normed spaces and can be constructed in 
a few equivalent ways.  In this work, we define our
LCTVS's via a family of seminorms.

\begin{definition}[Seminorm]
  Let $X$ be a real vector space, then 
  $q:X\rightarrow\real$ is a seminorm if 
  it satisfies the following properties.
  \begin{enumerate}
      \item $q(x)\ge0$ for all $x\in X$.
      \item $q(cx) = \abs{c}q(x)$ for all
        $x\in X$ and $c\in\real$.
      \item $q(x_1+x_2)\le q(x_1)+q(x_2)$.
  \end{enumerate}
\end{definition}

\begin{definition}[Locally Convex Space]
  A real vector space $X$ with $\mathcal{Q}$,
  a collection of seminorms, is said to be a 
  locally convex topological vector space.
  The seminorms in $\mathcal{Q}$ induce a 
  topology on $X$, which is the coarsest topology
  such that the $q\in\mathcal{Q}$ are continuous.
\end{definition}

An example of a LCTVS is $X = C_0([0,1],\real)$, the
space of continuous functions $x:[0,1]\rightarrow\real$,
$x(0)=0$
with $q_{\tau}(x) = \abs{x(\tau)}$ for $\tau\in[0,1]$.
More generally, we can consider $X^*$, the space of
linear functionals
on $X$, and take $q_f(x) = \abs{f(x)}$ for any $x\in X$
and $f\in X^*$.  Thus, we have a TVS $X$ with the weak
topology.

The focus of this work is on Gaussian measures 
on a LCTVS $X$. We say that $\gamma$ is a 
Gaussian measure defined on the cylindrical 
$\sigma$-field $\mathcal{E}(X)$ if the induced
measure $\gamma \circ f^{-1}$ on $\real$ is 
Gaussian for all $f\in X^*$.
We furthermore take $\gamma$ to be a Radon Gaussian 
measure, which is still sufficiently general 
for most applications of interest.  
\begin{defn}[Radon Measure]
  A measure $\mu$ defined on the Borel $\sigma$-field $\mathcal{B}(X)$
  for a topological space $X$ is Radon if 
  for every $B\in\mathcal{B}(X)$ and every $\veps>0$, there exists
  a compact set $K_\veps\subset B$ such that $\mu(B\setminus K_\veps)<\veps$.
\end{defn}
In the case that $X$ is a separable \frechet{} space,
the cylindrical $\sigma$-field coincides with the Borel
$\sigma$-field and furthermore every Borel measure is 
Radon;
see \cite{BOGACHEV1998}, Theorems~A.3.7 and~A.3.11.

We define similarly to \cite{BOGACHEV1998} the following
terms.  For a locally convex space $X$, $X^*$ is the topological 
dual space consisting of continuous linear functionals.  
The mean of $\gamma$ is $a_\gamma(f) = \int_X f(x)\gamma(dx)$ 
with $a_\gamma\in X^{**}$.  The covariance operator is
$R_\gamma: X^* \rightarrow X^{**}$ defined by 
$$
  R_\gamma(f)(g) = \int_X
  \left[f(x)-a_\gamma(f)\right]\left[g(x)-a_\gamma(g)\right]\gamma(dx).
$$
$X^*_\gamma$ is the closure of 
$\left\{f - a_\gamma(f)\,:\, f\in X^* \right\}$ embedded into $L^2(\gamma)$.
The Cameron-Martin space is
$$
 H(\gamma) = \left\{
   h \in X \,:\, \abs{h}_{H(\gamma)} < \infty
 \right\}
$$
where the norm is 
$\abs{h}_{H(\gamma)} = \sup\{l(h) \,:\, l\in X^*, R_\gamma(l)(l)\le1\}$.
Lemma~2.4.1 of \cite{BOGACHEV1998} proves that if some $h\in X$ is
in $H(\gamma)$, then there is a $h^*\in X_\gamma^*$ such that
$h = R_\gamma(h^*)$ with $\abs{h}_{H(\gamma)} = \norm{h^*}_{L^2(\gamma)}$.
Furthermore,
$\iprod{h}{k}_{H(\gamma)} = \iprod{h^*}{k^*}_{L^2(\gamma)}$.
Lastly, if $R_\gamma(X_\gamma^*) \subset X$, then 
$H(\gamma)=R_\gamma(X_\gamma^*)$ and 
$\abs{R_\gamma(f)}_{H(\gamma)} = \sqrt{ R_{\gamma}(f)(f) }$
where the operator $R_{\gamma}$ is 
extended to $X^*_\gamma$ as follows: 
$$
  R_{\gamma}: X_\gamma^*\to X^{**},~
  R_\gamma(f)(g) = \int_X
  f(x)\left[g(x)-a_\gamma(g)\right]\gamma(dx).
$$
This is necessarilly true for Radon Gaussian measures
\citep[Theorem 3.2.3]{BOGACHEV1998}.

\subsection{The Classic Hidden Markov Model}
\label{sec:standardHMM}

For the standard HMM, we begin with an observation 
sequence $O_1,\ldots,O_T$ that lives in some space $X$,
which could be a finite space $\{1,2,\ldots,d\}$ or 
Euclidean space $\real^d$ or otherwise.
In tandem, there is a hidden state sequence 
$s = (s_1,\ldots,s_T)$ where $s_t\in S = \{1,2,\ldots,p\}$,
which evolves as a $p$-state Markov chain with 
initial state probabilities 
$\eta_j = \prob{s_1 = j}$ and 
$p\times p$ transition matrix $A$ with $ij$th entry
$a_{ij} = \prob{ s_{t+1}=j \,|\, s_t = i }$, which is
assumed to be invariant to choice of $t=1,\ldots,T-1$.
However, nonhomogeneous Markov models with time varying
state transition probabilities have also been developed
\citep{HUGHES1999}.
Furthermore, there exist state dependent \textit{emission} 
functions $b_j:X\rightarrow\real^+$ for $j=1,\ldots,p$,
which assign a value to each observation $O_t$ based 
on being emitted from each potential state $s_t=j$.  
In the case of, say, multivariate Gaussian data in $\real^p$,
$b_j$ is simply the $d$-dimensional probability density 
function with state dependent mean vector and covariance 
matrix.
In the classic HMM, it is assumed that the observation 
sequence is comprised of $T$ elements that are independent
conditionally on the state sequence.  This assumption 
is removed in more complex variants of the HMM such 
as the autoregressive HMM discussed in \cite{RABINER1989}
and others.

The  Baum-Welch algorithm offers an efficient way 
to estimate the unknown parameters in the HMM that 
maximize the likelihood function 
$$
  L( O_1,\ldots,O_T \,|\, \lmb ) = 
  \sum_{s\in {S}^T} \eta_{s_1}
  \prod_{t=1}^T a_{s_ts_{t+1}} b_{s_t}(O_t)
$$
where the summation is taken over
${S}^T = \{1,\ldots,p\}^T$, the space of 
all $p^T$ state sequences,
and $\lmb$ represents the collection of model parameters.
The crux of the Baum-Welch algorthm are the forward and 
backward probabilities
\begin{align*}
  \alpha_t(j) &= \prob{ O_1,\ldots,O_t,s_t=j\,|\,\lmb }\\
  \beta_t(j) &= \prob{ O_{t+1},\ldots,O_T\,|\,s_t=j,\lmb },
\end{align*}
which can be computed recursively as outlined below in 
Algorithm~\ref{algo:baumWelch}.
Estimation of the state means is achieved as a weighted sum
of the observations where the weights come directly from the
$\alpha$ and $\beta$ probabilities.  Namely,
we wish to find the state means that maximize the following sum
$
       \sum_{t=1}^T 
       \alpha_t(j)\beta_t(j)
       b_j(O_t)
$
for each state $j$.
Our implementation of Baum-Welch is very similar to the classic
version.  The main innovation is usage and 
justification of the Onsager-Machlup
functional for the emission functions $b_j$ and thus maintaining
state dependent means that live in the Cameron-Martin space
$H(\gamma)$.

\subsection{Onsager-Machlup Functional}
\label{sec:onsager}

For a Gaussian measure $\gamma$ on a metric space $X$,
we can consider the \textit{Onsager-Machlup functional}, which is
$$
  I(a,b) = \lim_{\veps\rightarrow0}
  \frac{\gamma( K(a,\veps) )}{\gamma( K(b,\veps) )},
  ~~a,b\in X
$$
where $K(a,\veps)$ is a closed ball of radius $\veps>0$ centred
at $a\in X$.  Hence, we consider the limit as the radii of the two
balls shrink to zero. 
For a locally convex space, \cite{BOGACHEV1998} introduces the
notation $V_\veps = \left\{ x\in X\,:\,q(x)\le\veps \right\}$
where $q$ is a seminorm.
We are interested in
the ratio
$$
  \frac{\gamma(V_\veps+h)}{\gamma(V_\veps)} = 
  \frac{\ee^{-\abs{h}_H^2/2}}{\gamma(V_\veps)}
  \int_{V_\veps} \ee^{{h^*}(x)}\gamma(dx)
$$
where ${h}^*\in X_\gamma^*$ is such that $h = R_\gamma{h^*}$.
Furthermore, we denote the integral
$$
  J_\veps(f) = \frac{1}{\gamma(V_\veps)}
  \int_{V_\veps} \ee^{f(x)}\gamma(dx)
$$
and 
$
  F_q = \{ 
    f\in X_\gamma^* \,:\, \lim_{\veps\rightarrow0}J_\veps(f)=1 
  \}
$
is a closed linear subspace of $X_\gamma^*$ 
(see Lemma~4.7.2 in \cite{BOGACHEV1998}).
Via Lemma~4.7.4 \citep{BOGACHEV1998}, we can 
also define $F_q$ as $R_\gamma^{-1}Z^\perp$ where
$Z = \{ a\in H \,:\, q(a)=0 \}$.  Lastly, 
$P_q: X_\gamma^* \rightarrow F_q$ is an orthogonal projection.

Given this setup, the Onsager-Machlup function is
(see Corollary~4.7.8 \cite{BOGACHEV1998})
$$
  \lim_{\veps\rightarrow0}
  \frac{\gamma(V_\veps+h)}{\gamma(V_\veps+k)}
  = \exp\left(
    \frac{1}{2}\abs{\pi_qk}_H^2 - \frac{1}{2}\abs{\pi_qh}_H^2
  \right)
$$
for $h,k\in H(\gamma)$ where $\pi_q$ is the orthogonal 
projection onto $Z^\perp$, which can be written as 
$\pi_q = R_\gamma P_q R_\gamma^{-1}$.
We will use this for the emission function within the
HMM framework.  Namely, we choose
$$
  b_j(O_t) = 
  \lim_{\veps\rightarrow0}
  \frac{\gamma(V_\veps+\{O_t-h_j\})}{\gamma(V_\veps)}
  = \exp\left(
    -\frac{1}{2}\abs{\pi_q\{O_t-h_j\}}_H^2
  \right).
$$
Here, we require $O_t-h_j\in H(\gamma)$,
which may necessitate a modification of $O_t$
such as application of a kernel smoother.

\section{THMM Algorithms}
\label{sec:algorithm}

Borrowing the notation 
from classic works on HMMs 
\citep{LIPORACE1982,RABINER1989}, the goal for fitting an HMM 
to an observation sequence
$O = (O_1,\ldots,O_T)$
is to find the 
model parameters and the state
sequence $s$ that maximize the likelihood
$L(O\,|\,s) = \prod_{t=1}^T b_{s_t}(O_t)$.
The task of choosing the best parameters 
is achieved via the Baum-Welch algorithm.  
Determining the best state sequence is 
done by the Viterbi algoirthm.  These 
are detailed in Algorithms~\ref{algo:baumWelch}
and~\ref{algo:viterbi}, respectively, 
which differ only from their original 
instantiations in the choice of 
emission function $b_j$ and method of 
reestimation for the state means $h_j\in H(\gamma)$
for $j=1,\ldots,p$.

\subsection{Baum-Welch}
\label{sec:baumWelch}

\begin{algorithm}[t]
	\caption{
		\label{algo:baumWelch}
		The Baum-Welch Algorithm
	}
	\begin{tabbing}
	    \hspace{0.6\textwidth}\=\kill
		\qquad \enspace \bf Initialize model parameters:\\
		\qquad\qquad
		  $\eta_j = \prob{ s_1 = j }, \text{for }j=1,\ldots,p$
		  \>(Initial probabilities)\\
		\qquad\qquad
          $a_{ij} = \prob{ s_{t+1} = j \,|\, s_t = i }, \text{for }i,j=1,\ldots,p$
          \>(Transition probabilities)\\
        \qquad\qquad
          $h_j\in H(\gamma)$ for $j=1,\ldots,p$
          \>(Center element for each state)\\
        \qquad\enspace \bf
          Iterate until convergence:\\
        \qquad\qquad \enspace \bf Forward Pass:\\
        \qquad\qquad \qquad $\alpha_1(j) = \eta_j b_j(O_1)$\\
        \qquad\qquad \qquad For $t = 2,\ldots,T$:\\
        \qquad\qquad\qquad\qquad $\alpha_t(j) = 
            \prob{
              O_1,\ldots,O_t, s_t = j \,|\, \eta, A, h
            } = 
          \sum_{i=1}^p \alpha_{t-1}(i)a_{ij}b_j(O_t)$\\
        \qquad\qquad \enspace \bf Backward Pass:\\
        \qquad\qquad \qquad $\beta_T(j) = 1$\\
        \qquad\qquad \qquad For $t = (T-1),\ldots,1$:\\
        \qquad\qquad\qquad\qquad $\beta_t(j) =
            \prob{
              O_{t+1},\ldots,O_T \,|\, s_t = j, \eta, A, h
            } =
	        \sum_{i=1}^p \beta_{t+1}(i)a_{ji}b_i(O_{t+1})$\\
	    \qquad\qquad \enspace \bf Reestimation:\\
	    \qquad\qquad\qquad
	      $
	        \gamma_t(i) = \prob{
              s_t = i \,|\, O, \eta, A, h
            }
            = \frac{
              \alpha_t(i)\beta_t(i)
            }{
              \sum_{j=1}^p \alpha_t(j)\beta_t(j)
            }
        $\\
        \qquad\qquad\qquad
        $
        \xi_t(i,j) = \prob{
          s_t = i, s_{t+1}=j \,|\,
          O, \eta, A, h
        } 
        = \frac{
          \alpha_t(i) a_{ij} b_j(O_{t+1})\beta_{t+1}(j)
        }{
          \sum_{i',j'=1}^p 
          \alpha_t(i')a_{i'j'}b_{j'}(O_{t+1})\beta_{t+1}(j')
        }
        $\\
        \qquad\qquad\qquad
        $\tilde{\eta}_j  = \gamma_1(j)$\\
        \qquad\qquad\qquad
        $
          \tilde{a}_{ij} = \frac{
            \sum_{t=1}^{T-1} \xi_t(i,j)
          }{
            \sum_{t=1}^{T-1} \gamma_t(i)
          }
        $\\
        \qquad\qquad\qquad
        $h_j = \argmax{h\in H(\gamma)} 
          \sum_{t=1}^T 
          \alpha_t(j)\beta_t(j)
          b_j(O_t)$ where 
          $b_j(O_t)$ depends on $h$.
	\end{tabbing}
\end{algorithm}

The Baum-Welch Algorithm
\citep{BAUM1966}, detailed in 
Algorithm~\ref{algo:baumWelch}, takes a form 
similar to that of an Expectation-Maximization 
algorithm.  However, the Baum-Welch algorithm 
predates the EM algorithm \citep{DEMPSTER1977} 
by about a decade.
For more classic references on HMMs, see those
within \cite{RABINER1989}.

The Baum-Welch algorithm works by computing so-called
forward (alpha) and backward (beta) probabilities
given the model parameters.  It then reestimates the
the model parameters based on these probabilities.
At each iteration, the likelihood increases and the
algorithm is run until the relative change in the 
likelihood becomes minuscule. 
Given the alpha probabilities, the log likelihood is 
simply computed as 
$\ell(O) = \sum_{j=1}^p \log \alpha_T(j)$.
Denoting the log likelihood at iteration $r$ to be 
$\ell_r$, we stop the algorithm when the relative
change in the log likelihood, $(\ell_{r+1}-\ell_r)/\ell_{r+1}$,
is less than a user specified tolerance, say, $10^{-6}$.
In practice, all of the terms in Algorithm~\ref{algo:baumWelch}
are computed on the log-scale to avoid numerical 
stability issues.

Reestimation of the means is performed by
$$
h_j = \argmax{h\in H(\gamma)} 
  \sum_{t=1}^T 
  \alpha_t(j)\beta_t(j)
  b_j(O_t)
$$
where the emission probabilities, $b_j(O_t)$, depend on 
choice of state mean $h_j$.  How this equation is used
depends on the type of model being fitted.  Many specific
examples are considered below in Section~\ref{sec:interestingSpaces}.

As with both the classic HMM and EM-style algorithms,
the choice of initialization parameters can drastically
affect the performance.
In particular, we require each state to have a mean 
$h_j\in H(\gamma)$, which furthermore lies within the
convex hull of the observations $O_1,\ldots,O_T$.  This
will be discussed more in Section~\ref{sec:theory}.
For parametric settings, we estimate the parameters
for each $O_t$ and then pick random starting values 
roughly spread out in this convex set.  Alternatively, 
in the non-parametric setting, we can either 
randomly select a single $O_t$ for each state to start
with or we can run a quick k-means clustering to 
automatically choose the starting state means.

\subsection{Viterbi}
\label{sec:viterbi}

\begin{algorithm}[t]
	\caption{
		\label{algo:viterbi}
		The Viterbi Algorithm
	}
	\begin{tabbing}
		\qquad \enspace Initialize
		$\delta_1(j) = \eta_j b_j(O_1) \text{ and }
         \phi_1(j) = 0 \text{ for all }i=1,\ldots,p$.\\
		\qquad \enspace 
		For $t=2,\ldots,T$, compute the following 
		for all $j=1,\ldots,p$\\
		\qquad\qquad $
		  \delta_t(j) = \max_{i=1,\ldots,p}\left\{ 
            \delta_{t-1}(i)a_{ij} 
          \right\}b_j(O_t)
        $\\
        \qquad\qquad $
          \phi_t(j) = \argmax{i=1,\ldots,p}\left\{
            \delta_{t-1}(i)a_{ij}
          \right\}
        $\\
        \qquad\enspace
        The final state is $\hat{s}_T =
        \arg\max_{i=1,\ldots,p}\delta_T(i)$.\\
        \qquad \enspace 
		For $t=(T-1),\ldots,1$, compute\\
        \qquad\qquad
        $
         \hat{s}_t = \phi_{t+1}(\hat{s}_{t+1}).
        $
	\end{tabbing}
\end{algorithm}

Given an observation sequence $O_1,\ldots,O_T$,
a state space, initial probabilities $\eta_j$, 
transition probabilities $a_{ij}$, and
emission probabilities $b_j(O_t)$, 
the Viterbi algorithm \citep{VITERBI1967} 
finds the most
probable state sequence;  
see \cite{RABINER1989} and 
references therein for more details.

Let 
$
 \delta_t(j) = 
 \max_{(s_1,...,s_{t-1})\in S^{t-1}} P(s_1,...,s_{t-1},s_t=j)
$
be the highest probability
of any state sequence from $1$ to $t$ such that 
the state at time $t$ is $j$.
The probability of the most likely state sequence 
ending in state $s_t=j$ can be computed in a 
recursive fashion by maximizing over $s_{t-1}$ 
at each time step.
Let 
$\phi_t(j)$ be the state at time $t-1$ 
that maximizes $\delta_t(j)$.
The goal of this algorithm is to compute the 
\textit{best} state sequence denoted 
$\hat{s}_1,\ldots,\hat{s}_T$.
The Viterbi algorithm is detailed in 
Algorithm~\ref{algo:viterbi}.

\section{Spaces of Interest}
\label{sec:interestingSpaces}

\subsection{Euclidean Space}
\label{sec:OMEuclid}

For the simplest setting,
we can consider 
$O_t\in\real^d$ and each state $s_t=j$ 
corresponding to a multivariate
Gaussian distribution with mean $\mu_j$ and 
common covariance $\Sigma$.  In this case,
the Cameron-Martin norm is 
$
  \abs{h}_{H(\gamma)} = 
  \sup\{
    \TT{v}h \,:\, \TT{v}\Sigma v \le 1
  \}
$
where $\xv(v) =0$ and $\var{v} = \TT{v}\Sigma v$.  
This leads to 
$
  \abs{h}_{H(\gamma)} = 
  \sqrt{ \TT{h}\Sigma^{-1}h }.
$
The corresponding Cameron-Martin inner product is
$
  \iprod{h}{k}_{H(\gamma)} = 
  { \TT{h}\Sigma^{-1}k }.
$
Thus, the log-emission function is 
$\log b_j(O_t) = -\frac{1}{2}\TT{(O_t-h)}\Sigma^{-1} (O_t-h)$
and the reestimated mean vector is
$$
  \tilde{m}_j = \argmin{h\in \real^d} 
  \sum_{t=1}^T 
  \alpha_t(j)\beta_t(j)
  \TT{(O_t-h)}\Sigma^{-1} (O_t-h),
$$
which can be simply solved via vector calculus to get
$$
  \tilde{m}_j = \frac{
    \sum_{t=1}^T \alpha_t(j)\beta_t(j)O_t
  }{
    \sum_{t=1}^T \alpha_t(j)\beta_t(j)
  }.
$$ 
This formulation coincides with the classic 
HMM for multivariate Gaussian data when the 
covariance matrix is fixed.  Our current 
formulation of the HMM using the Onsager-Machlup
functional requires a fixed covariance structure 
across all system states.  However, as we will 
show below, this still allows the THMM to
model many diverse types of data.  

Theoretical justification for the reestimated 
means can be found in Lemma~\ref{lem:meanreest}.
In this and the following specific cases, the
reestimated means take the form of weighted 
averages of the observed data points.

\subsection{Wiener Space with Smooth Norms}
\label{sec:OMWiener}

In a collection of articles and texts
\citep{TAKAHASHI1981,ZEITOUNI1989,SHEPP1992,CAPITAINE1995,IKEDA2014},
the Onsager-Machlup functional
is derived for the diffusion process 
$$
  dY_\tau = r(Y_\tau) d\tau + dW_\tau,~~
  Y_0 = y, ~~ Y_\tau \in \real^d, ~~~\tau\in[0,1],
$$
where $r:\real^d\rightarrow\real^d$ is a smooth function 
and $W_\tau$ is $d$-dimensional Brownian motion.
That is,
$$
	\log\left[
	  \lim_{\veps\rightarrow0} 
	  \frac{\prob{\norm{Y-\Phi}<\veps}}{\prob{\norm{W}<\veps}}
	\right]
	=
	-\frac{1}{2} \sum_{k=1}^d \int_0^1
	\abs*{ \dot{\Phi}_{k,\tau} -  r_k(\Phi_\tau) }^2 d\tau
	-\frac{1}{2}\sum_{k=1}^d\int_0^1 
	\frac{\partial r_k}{\partial y_k}(\Phi_\tau)d\tau
$$
where $\dot{\Phi}_\tau = d\Phi_\tau/d\tau$.
\cite{IKEDA2014} prove the above for the sup-norm
and $\Phi\in C^2$, the space of twice differentiable
functions.
\cite{SHEPP1992} extended this to all $\Phi$ such 
that $\Phi-y\in H(\gamma)$ and $L^p$ norms for 
$p\ge4$ and \holder{} norms with $0<\alpha<1/3$.
\cite{CAPITAINE1995} further shows that this result
holds for a wide class of smooth norms on Wiener 
space including \holder{} norms with $0<\alpha<1/2$,
Besov norms, and Sobolev norms.

\subsubsection*{Brownian Motion with Drift}

Many common stochastic processes arise from choices 
of $r$ 
(see \cite{PAVLIOTIS2014} Section~5.3).
For example, in the case of one-dimensional Brownian motion 
with state $j$ drift coefficient $c_j$, the diffusion
equation is $dY_{\tau}=c_j d\tau+dW_{\tau}$, and
the Onsager-Machlup functional / log-emission function
is 
$$
  \log b_j(O_{t,\tau}) = 
  \log\left[
    \lim_{\veps\rightarrow0} 
    \frac{\prob{\norm{W_\tau+c_j\tau-O_{t,\tau}}<\veps}
    }{\prob{\norm{W_\tau}<\veps}}
  \right]
  =
  -\frac{1}{2} \int_0^1
  \abs*{ \dot{O}_{t,\tau} -  c_j }^2 d\tau.
$$
The observations $O_{t,\tau}$ should be projected into the 
Cameron-Martin space $H(\gamma)$ allowing for differentiation.
In practice, smoothing methods may be required.
The drift terms can be reestimated within the Baum-Welch 
algorithm by
$$
  \tilde{c}_j = \argmin{h\in \real} 
  \sum_{t=1}^T \alpha_t(j)\beta_t(j)
  \left\{
    \frac{1}{2} \int_0^1 \abs*{ \dot{O}_{t,\tau} -  h }^2 d\tau
  \right\}.
$$
This leads to the weighted least squares estimate
$$
  \tilde{c}_j = 
  \frac{
  	\sum_{t=1}^T \alpha_t(j)\beta_t(j)
  	[O_{t,1} - O_{t,0}]
  }{
    \sum_{t=1}^T \alpha_t(j)\beta_t(j)
  }.
$$

Restricting to a one-dimensional drift parameter
is convenient for exposition, but not necessary
in practice.  We could instead consider
$$
  dY_\tau = r(t) d\tau + dW_\tau,~~
  Y_0 = y, ~~ Y_\tau \in \real^d, ~~~\tau\in[0,1],
$$
where $r(t) = \sum_{i=1}^m c_{i}\psi_i(t)$ for
some functional basis $\psi_1,\ldots,\psi_m$.
Thus, we would have $m$ parameters to estimate.

\subsubsection*{Ornstein-Uhlenbeck Process}

In the case of the one-dimensional 
Ornstein-Uhlenbeck (OU) process,
$dY_{\tau}=c_j(\mu_j-Y_{\tau})d{\tau}+dW_{\tau}$, 
the Onsager-Machlup functional is 
\begin{multline*}
    \log b_j(O_{t,\tau}) =
	\log\left[
	  \lim_{\veps\rightarrow0} 
	  \frac{
	    \prob{\norm{W_\tau+c_j(\mu_j-Y_\tau)-O_{t,\tau}}<\veps}
	  }{
	    \prob{\norm{W}<\veps}
	  }
	\right]\\
	=
	-\frac{1}{2}  \int_0^1
	\abs*{ \dot{O}_{t,\tau} -  c_j(\mu_j-O_{t,\tau}) }^2 d\tau
	+\frac{c_j}{2},
\end{multline*}
which leads to the following parameter 
reestimation:
\begin{align*}
  (\tilde{c}_j,\tilde{\mu}_j) &= \argmin{h\in \real^+,k\in\real} 
  \sum_{t=1}^T \alpha_t(j)\beta_t(j)
  \left\{
    \frac{1}{2}  \int_0^1
	\abs*{ \dot{O}_{t,\tau} -  h(k-O_{t,\tau}) }^2 d\tau
	-\frac{h}{2}
  \right\}.
\end{align*}
The derivative with respect to $k$ inside the curly brackets
gives similarly to the above Brownian motion that
\begin{align*}
  \tilde{\mu}_j &= 
  \frac{
	\sum_{t=1}^T \alpha_t(j)\beta_t(j)
	\int_0^1
	\{ h\dot{O}_{t,\tau} +  h^2O_{t,\tau} \} d\tau
  }{
	\sum_{t=1}^T \alpha_t(j)\beta_t(j)h^2
  }\\
  &= 
  \frac{
	\sum_{t=1}^T \alpha_t(j)\beta_t(j)
	\{ {O}_{t,1}-{O}_{t,0} 
	+  h\int_0^1 O_{t,\tau}d\tau \} 
  }{
	\sum_{t=1}^T \alpha_t(j)\beta_t(j)h
  }.
\end{align*}
The derivative with respect to $h$ inside the curly brackets
gives
\begin{multline*}
  \int_{0}^1 \left[
    \dot{O}_{t,\tau} -  h(k-O_{t,\tau})
  \right](O_{t,\tau}-k)
  d\tau
  - \frac{1}{2} 
  \\ = 
    -k( {O}_{t,1}-{O}_{t,0} )
    + \int_{0}^1\dot{O}_{t,\tau}O_{t,\tau}d\tau 
    + h\int_{0}^1(k-O_{t,\tau})^2d\tau
  - \frac{1}{2}. 
\end{multline*}
Thus,
$$
  \tilde{c}_j = 
  \frac{
  	\sum_{t=1}^T \alpha_t(j)\beta_t(j)\left[
  	  \frac{1}{2}+k( {O}_{t,1}-{O}_{t,0} )
      - \int_{0}^1\dot{O}_{t,\tau}O_{t,\tau}d\tau 
  	\right]
  }{
    \sum_{t=1}^T \alpha_t(j)\beta_t(j)
    \left[
      \int_{0}^1(k-O_{t,\tau})^2d\tau
    \right]
  }.
$$
In practice, we reparametrize the OU process
to stabilize this nonlinear optimization over the two parameters,
which is discussed in Section~\ref{sec:OUP}.

\subsection{Fractional Brownian Motion}
\label{sec:OMFractBM}

In \cite{MORET2002}, the Onsager-Machlup functional is
derived for fractional Brownian motion in the 
\textit{Singular Case}, which is where the Hurst
parameter is $\frac{1}{4}<\nu<\frac{1}{2}$, and
the \textit{Regular Case} where the Hurst parameter
is $\nu>\frac{1}{2}$.  Note that for $\nu=\frac{1}{2}$,
we have standard Brownian motion.  The process considered
is
$$
  Y_\tau = y + W_\tau^\nu + \int_0^\tau r(Y_s)ds
$$
where $r\in C_b^2(\real)$, the space of bounded functions with two continuous
derivatives.

The Onsager-Machlup functional for the singular
case from Theorem~7 in \cite{MORET2002} is 
$$
	\log\left[
	\lim_{\veps\rightarrow0} 
	\frac{\prob{\norm{Y-\Phi}<\veps}}{\prob{\norm{W^\nu}<\veps}}
	\right]
	=
	-\frac{1}{2}\int_0^1\left\{
	  \dot{\Phi}_\tau - \tau^{-\upsilon}I_{0^+}^\upsilon 
	  \tau^\upsilon r(\Phi_\tau)
	\right\}^2d\tau
	-\frac{1}{2} d_\nu\int_0^1 r'(\Phi_\tau)d\tau
$$
where $K^\nu\dot{\Phi} = \Phi - x$, $\upsilon = \abs{\nu-1/2}$,
$K^\nu$ is the operator such that $dW_t^\nu = K^\nu(t,s)dW_s$,
$$
  d_\nu = \sqrt{
  	\frac{2\nu\Gamma(3/2-\nu)\Gamma(\nu+1/2)}{\Gamma(2-2\nu)}
  },
$$
and 
$
  I_{a^+}^\upsilon f(x) = \Gamma(\upsilon)^{-1}\int_a^x
  (x-y)^{\upsilon-1}f(y)dy
$
is called the \textit{left fractional Riemann-Liouville Integral}.

In the simplest non-trivial setting of 
$r = c\in\real$ and fixed 
$\upsilon\in(0,1/4)$, we aim to solve for the drift term
$c$ such that
$$
  \tilde{c}_j = \argmin{h\in \real} 
  \sum_{t=1}^T \alpha_t(j)\beta_t(j)
  \left\{
    \frac{1}{2}\int_0^1\left\{
	  \dot{\Phi}_\tau - h\tau^{-\upsilon}I_{0^+}^\upsilon 
	  \tau^\upsilon 
	\right\}^2d\tau
  \right\}.
$$
In this case, the integral $I_{0^+}^\upsilon\tau^\upsilon$
is a scaled Beta function and
$
  \tau^{-\upsilon}I_{0^+}^\upsilon\tau^\upsilon = 
  \tau^{\upsilon}\Gamma(\upsilon+1)/\Gamma(2\upsilon+1).
$
Thus, some simple calculus results in 
$$
  \tilde{c}_j =
  \frac{\Gamma(2\upsilon+2)}{\Gamma(\upsilon+1)}
  \frac{
  	\sum_{t=1}^T \alpha_t(j)\beta_t(j)
  	\int_{0}^1 \tau^\upsilon \dot{O}_{t,\tau} d\tau
  }{
    \sum_{t=1}^T \alpha_t(j)\beta_t(j)
  }.
$$
Setting $\upsilon=0$ in the above returns us to the 
formula for $\tilde{c}_j$ derived from Brownian motion
with drift.

Similarly, the Onsager-Machlup functional for the regular
case, $\nu>1/2$, from Theorem~8 in \cite{MORET2002} is 
$$
	\log\left[
	\lim_{\veps\rightarrow0} 
	\frac{\prob{\norm{Y-\Phi}<\veps}}{\prob{\norm{W^\nu}<\veps}}
	\right]
	=
	-\frac{1}{2}\int_0^1\left\{
	  \dot{\Phi}_\tau - \tau^{\omega}D_{0^+}^{\omega} 
	  \tau^{-\omega} r(\Phi_\tau)
	\right\}^2d\tau
	-\frac{1}{2} d_\nu\int_0^1 r'(\Phi_\tau)d\tau
$$
for $\omega = \nu - 1/2$ and $D_{0^+}^\omega$
is the left-sided Riemann-Liouville derivative
defined as 
$$
  D_{a^+}^\omega f(x) = \frac{1}{\Gamma(1-\omega)}
  \frac{d}{dx}\int_{a}^x \frac{f(y)}{(x-y)^{\omega}}dy.
$$
If we consider the linear drift setting of 
$r = c \in\real$, then 
$
  \tau^{\omega}D_{0^+}^\omega \tau^{-\omega}
  = (1-2\omega)\tau^{-\omega}
  \Gamma(1-\omega)/\Gamma(2-2\omega).
$
A little more calculus gives us
$$
  \tilde{c}_j =
  \frac{\Gamma(2-2\omega)}{\Gamma(1-\omega)}
  \frac{
  	\sum_{t=1}^T \alpha_t(j)\beta_t(j)
  	\int_{0}^1 \tau^{-\omega} \dot{O}_{t,\tau} d\tau
  }{
    \sum_{t=1}^T \alpha_t(j)\beta_t(j)
  },
$$
which coincides nicely with the singular case 
as $-\omega = \upsilon$.

Applying our THMM algorithm to fractional Brownian
motion extends the range of possible stochastic 
processes we can consider.  When the Hurst parameter
$\nu>1/2$, the process has positively correlated 
increments and thus appears smoother than standard
Brownian motion.  In comparison, processes with 
$\nu<1/2$ exhibit negatively correlated increments
and thus appear rougher than standard Brownian
motion.

\subsection{Non-Parametric State Means}
\label{sec:nonpar}

The previous sections consider estimation of 
specific real valued parameters under different 
stochastic models.  
However, a more flexible 
approach is to treat estimation of the means 
non-parametrically.
This is achieved by using the formulae derived from the
Onsager-Machlup functional at the beginning of
Section~\ref{sec:onsager} directly. 

Indeed, if $q$ is a norm, then $\pi_q$ is the identity.
The emission function and mean 
for state $j$ are shown to be
$$
  b_j(O_t) = 
  \exp\left(
    -\frac{1}{2}\abs{O_t-h_j}_H^2
  \right) ~\text{ and }~
  h_j = \frac{
    \sum_{t=1}^T\alpha_t(j)\beta_t(j)O_t
  }{
    \sum_{t=1}^T\alpha_t(j)\beta_t(j)
  },
$$
respectively.  To see the latter equation, we note
that 
\begin{align*}
  \sum_{t=1}^T
  \alpha_t(j)\beta_t(j)
  \abs{ O_t-h_j }_H^2
  &= 
  \sum_{t=1}^T
  \alpha_t(j)\beta_t(j)
  R_\gamma( O_t^*-h_j^* )( O_t^*-h_j^* ) \\
  &= 
  \sum_{t=1}^T
  \alpha_t(j)\beta_t(j)
  \int_X ( O_t^*-h_j^* )^2 \gamma(dx),
\end{align*}
which is minimized by
$
  h_j^* = {
    \sum_{t=1}^T\alpha_t(j)\beta_t(j)O_t^*
  }/{
    \sum_{t=1}^T\alpha_t(j)\beta_t(j)
  }.
$
Applying the operator $R_\gamma$ to each side recovers
the optimal $h_j$.

In practice, one must select a suitable Gaussian measure
/ Cameron-Martin space for the problem at hand.
For example, in Section~\ref{sec:pedOSA}, we analyze a sequence of
Electroencephalogram (EEG) signals under the standard
Wiener measure.  In this setting, $H=W_0^{2,1}[0,1]$,
the Sobolev space of absolutely continuous functions
$h$ such that $h'\in L^2[0,1]$ with $h(0)=0$.
As a consequence, the emission function becomes 
$$
   b_j(O_t) = 
  \exp\left(
    -\frac{1}{2}\int_{0}^1 
    \abs*{\dot{O}_t(\tau)-\dot{h}_j(\tau)}^2
    d\tau
  \right).
$$
However, it is worth emphasizing that other 
Cameron-Martin norms can be considered and may 
improve performance of algorithm.

\section{Theoretical Guarantees}
\label{sec:theory}

\subsection{Reestimation and Maximum Likelihood}

In this section, we extend proofs from past
work \citep{LIPORACE1982,wu1983convergence} to show 
that the Baum-Welch and EM algorithms 
{satisfy nice convergence properties}.  
Given a finite sequence of observations
$O=\{O_t\}_{t=1}^T$, initial state probabilities
$\{\eta_j\}_{j\in S}$, a $p\times p$ Markov transition matrix $A$ 
with $ij$th entry $a_{ij}$, state means $\{m_j\}_{j\in S}$ and a state sequence $s = (s_1,\ldots,s_T)\in S^T$, we 
can define our analogue of the likelihood function
(refer to Section~\ref{sec:standardHMM} for the classic 
HMM setup)
to be 
$$
  L_{\lmb}(O) = \sum_{s\in S^T}
  \left(
    \eta_{s_1}\prod_{t=1}^T a_{s_{t-1}s_t}
  \right)
  \exp\left\{
    -\frac{1}{2}\sum_{t=1}^T\abs{
	  \pi_{q}(O_t-m_{s_t})
    }_H^2 
  \right\}
$$
where the sum is taken over all state sequences and  
${\lmb}=(\{{\eta}_j\}_{j\in S}$, $\{{a}_{ij}\}_{i,j\in S}$,
$\{{m}_j\}_{j\in S})$ denotes the selection of parameters. 
Let $\Lambda$ be the space of all possible parameters 
$\lambda$. The initial probabilities $\{{\eta}_j\}_{j\in S}$ 
and each row of $\{{a}_{ij}\}_{i,j\in S}$ lie in 
the $p-1$ simplex, i.e. the closed convex hull of 
the unit vectors in $\real^p$, which is compact. 
Furthermore, each $m_j$ for ${j\in S}$ lies in 
${H}_0$, 
a convex subset of the Cameron-Martin space $H(\gamma)$.
Thus, $\Lambda$ is a closed convex subset of $\mathbb{R}^p\times \mathbb{R}^{p\times p}\times H(\gamma)^p$.
For a specific state sequence $s\in S^T$, we write
$L_\lmb(O,s)$ to be the summand of $L_\lmb(O)$
for $s$.
The reestimation transformation 
$Q:\Lambda\times\Lambda\rightarrow\real$ is a bivariate 
function given by
\begin{align*}
  Q(\lmb,\tilde{\lmb})
  &= \sum_{s\in S^T}
  L_\lmb(O,s)\log L_{\tilde{\lmb}}(O,s)\\
  &= \sum_{s\in S^T}\left[
  L_\lmb(O,s)\left\{
    \log \tilde{\eta}_{s_1} + 
    \sum_{t=1}^T \log \tilde{a}_{s_{t-1}s_t} -
    \frac{1}{2}\sum_{t=1}^T\abs{  
	  {\pi}_{q}(O_t-\tilde{m}_{s_t})
    }_H^2
  \right\}
\right].
\end{align*}
The Baum-Welch algorithm is part of the family of 
majorize-minimization algorithms that optimize this function in an iterative fashion in order to obtain parameter estimates, instead of 
directly optimizing the likelihood.
Our goal in this section is to prove that maximizing
$Q(\lmb,\tilde{\lmb})$ over all $\tilde{\lmb}$ increases
the likelihood, i.e.  $L_{\lmb}(O) \le L_{\tilde{\lmb}}(O)$, and that the
the reestimation procedure stabilizes only at critical points 
of the likelihood. When talking about differentiability in 
this paper, we will always be referring to the \frechet{}
derivative. 

\begin{lemma}
The likelihood function $L(\lambda)=L_\lmb(O)$ and the reestimation function $Q$ are differentiable with respect to the state means.
\end{lemma}

\begin{proof}

Firstly, we note that the squared norm on real Hilbert spaces is differentiable, since 
 $$\|x+h\|^2=\langle x+h, x+h \rangle
= \|x\|^2+\|h\|^2+2\langle x,h \rangle = \|x\|^2 + 2\langle x,h \rangle +O(h)$$ is linear in $h$. Differentiability of both functions in $\{m_j\}_{j\in S}$ then follows from this and the smoothness of projections \cite[Corollary 6.2]{coleman2012calculus}. 
\end{proof}

\begin{lemma} \label{lem:meanreest} The function
	$
	\psi(\{\tilde{m}_j\}_{j\in S}) = 
	\sum_{s\in S^T} L_\lmb(O,s) \sum_{t=1}^T
	\abs{  
		{\pi}_{q}(O_t-\tilde{m}_{s_t})
	}_H^2
	$ 
	has global minima for each 
	$$
	{\tilde{m}}_j\in
	 \left\{z+\frac{\sum_{t=1}^T 
		\alpha_t(j)\beta_t(j){\pi}_{q}(O_t)}{\sum_{t=1}^T 
		\alpha_t(j)\beta_t(j)} \,:\, z\in Z\right\}
	$$
	where $Z = \{ a\in H \,:\, q(a)=0 \}$.
	Furthermore, these are the only critical points of the function.
\end{lemma}

\begin{proof}
	Let us denote ${T}_j(s) = \{t\,:\,s_t=j\}\subseteq\{1,\ldots,T\}$, 
	$S_j(t) = \{s\,:\,s_t=j\}\subset S^T$, 
	and  assume the forward  and backward probabilities 
	are known at the current step, 
	denoted by $\alpha_t$ and $\beta_t$ respectively. 
	We may then rewrite our function $\psi$ as follows.
	\begin{align*}
	  \psi(\{\tilde{m}_j\}_{j\in S})
	  &= \sum_{s\in S^T}
	     L_\lmb (O,s) \sum_{j=1}^p\sum_{t\in T_j(s)}\abs{  
		  {\pi}_{q}(O_t-\tilde{m}_{j})
	    }_H^2 \\
	 &=
	 \sum_{j=1}^p \sum_{t=1}^T \sum_{s\in S_j(t)} 
	   L_\lmb(O,s) \abs{  
	     {\pi}_{q}(O_t-\tilde{m}_j)
	   }_H^2 \\
	&=
	\sum_{j=1}^p
	\sum_{t=1}^T 
	\alpha_t(j)\beta_t(j)
	\abs{  
	  {\pi}_{q}(O_t-\tilde{m}_j)
	}_H^2.
	\end{align*}
	Each of the $p$ terms in the outer sum is a non-negative 
	function of only one of the parameters $\tilde{m}_j$ 
	to be optimized. We may therefore separately optimize each 
	of the terms. Abusing notation and letting $\tilde{m}_j$ 
	represent the set of minimizers for $\psi$, 
	we have 
	$$ 
	  \tilde{m}_j= \argmin{m\in H} \sum_{t=1}^T 
      \alpha_t (j)\beta_t (j)\abs{{\pi}_{q}(O_t-m)}_H^2.
    $$ 
    We note here that the arg\,min is invariant under 
    translation by elements in $Z$.
    For $B_1(0)\subset X$, the ball of radius 1 about the origin,
	{\allowdisplaybreaks
	\begin{align*}
	\pi_q(\tilde{m}_j)&=\argmin{h\in Z^\perp} 
	\sum_{t=1}^T 
	\alpha_t(j)\beta_t(j)
	\abs{  
		{\pi}_{q}(O_t)-h
	}_H^2\\
	&= \argmin{\substack{h_0\in Z^\perp\cap B_1(0)\\c>0}} 
	\sum_{t=1}^T 
	\alpha_t(j)\beta_t(j)
	\abs{{\pi}_{q}(O_t)-ch_0}_H^2   \\
	& = \argmin{\substack{h_0\in Z^\perp\cap B_1(0)\\c>0}} 
	\sum_{t=1}^T 
	\alpha_t(j)\beta_t(j)\left\{ \abs{{\pi}_{q}(O_t)}_H^2+c^2\abs{h_0}_H^2 -2 c\iprod{{\pi}_{q}(O_t)}{h_0}_H  \right\}\\
	& = \argmin{\substack{h_0\in Z^\perp\cap B_1(0)\\c>0}} 
	\sum_{t=1}^T 
	\alpha_t(j)\beta_t(j)\left\{c^2 -2c\iprod{{\pi}_{q}(O_t)}{h_0}_H \right\}\\
	&=\argmin{\substack{h_0\in Z^\perp\cap B_1(0)\\c>0}}\left\{c^2\sum_{t=1}^T 
	\alpha_t(j)\beta_t(j) -2c\iprod{\sum_{t=1}^T 
		\alpha_t(j)\beta_t(j){\pi}_{q}(O_t)}{h_0}_H \right\}.
	\end{align*}
	}
	
	Now we note that the variable $h_0$ minimizing the above equation 
	is independent of the choice of $c$ and that the second term is minimized 
	exactly when $h_0$ is parallel to the term in the inner product,
	i.e. $h_0= {\sum_{t=1}^T 
	\alpha_t(j)\beta_t(j){\pi}_{q}(O_t)}{|\sum_{t=1}^T 
	\alpha_t(j)\beta_t(j){\pi}_{q}(O_t)|_H^{-1}}$. We therefore have 
	\begin{multline*}
	\argmin{\substack{h_0\in Z^\perp\cap B_1(0)\\c>0}}\left\{c^2\sum_{t=1}^T 
	\alpha_t(j)\beta_t(j) -2c\iprod{\sum_{t=1}^T 
		\alpha_t(j)\beta_t(j){\pi}_{q}(O_t)}{h_0}_H \right\}\\
	= \argmin{c>0}\left\{c^2\sum_{t=1}^T 
	\alpha_t(j)\beta_t(j) -2c\left|\sum_{t=1}^T 
	\alpha_t(j)\beta_t(j){\pi}_{q}(O_t)\right|_H\right\}.
	\end{multline*}
    Thus we have 
    $$
    c=\frac{
      \abs{\sum_{t=1}^T 
	  \alpha_t(j)\beta_t(j){\pi}_{q}(O_t)}_H
	}{\sum_{t=1}^T \alpha_t(j)\beta_t(j)}
	~\text{ and }~
	\pi_q(\tilde{m}_j)=ch_0=\frac{
	  \sum_{t=1}^T \alpha_t(j)\beta_t(j){\pi}_{q}(O_t)
	}{
	  \sum_{t=1}^T \alpha_t(j)\beta_t(j)
	}.
	$$
	It follows that the members of 
	the product of sets 
	$\{\pi_q(\tilde{m}_j)+Z\}$ over $j\in S$ 
	are the global minima of $\psi$ over the open convex set
	${H}_0$, and thus, they are
	local minima, and hence, they are  
	also critical points \cite[Corollary 2.5]{coleman2012calculus}. 
	The lack of other critical points follows from 
	the convexity of $\psi$ over $H_0^{p}$ owing to the 
	convexity of the square norm and linearity of 
	$\pi_q$. See Theorem 7.4 (c) and Proposition 7.4 from \cite{coleman2012calculus}.
\end{proof}

\begin{theorem} \label{thm:globalmax}
	Every critical point of the reestimation function $Q(\lmb,\cdot)$ 
	is a global maximum and at least one such point exists. 
	Additionally, if $q$ is a norm, this maximum is unique.
\end{theorem}

\begin{proof}
	
The reestimation function $Q(\lmb,\tilde{\lmb})$ can be written as 
	\begin{align*}
	Q(\lmb,\tilde{\lmb})
	=& \sum_{s\in S^T}\left[
	L_\lmb(O,s)\left\{
	\log \tilde{\eta}_{s_1} + 
	\sum_{t=1}^T \log \tilde{a}_{s_{t-1}s_t} -
	\frac{1}{2}\sum_{t=1}^T\abs{  
		{\pi}_{q,s_t}(O_t-\tilde{m}_{s_t})
	}_H^2
	\right\}
	\right]\\
	=& \underbrace{\sum_{s\in S^T}
		L_\lmb(O,s)\log \tilde{\eta}_{s_1}}_{(1)}+
	   \underbrace{\sum_{s\in S^T}
		L_\lmb(O,s)	\sum_{t=1}^T \log \tilde{a}_{s_{t-1}s_t} }_{(2)}\\
	&-  \underbrace{ \frac{1}{2}\sum_{s\in S^T}\sum_{t=1}^T
		L_\lmb(O,s)
		\abs{  
			{\pi}_{q}(O_t-\tilde{m}_{s_t})
		}_H^2}_{(3)}.
	\end{align*}
	
	The arg\,max of the above 
	in $\tilde{\lmb}=(\{\tilde{\eta}_j\}_{j\in S}$, $\{\tilde{a}_{ij}\}_{i,j\in S}$,
	$\{\tilde{m}_j\}_{j\in S})$ can be broken down into taking the arg\,max of 
	(1) and (2) and the arg\,min of (3), 
	in $\{\tilde{\eta}_j\}_{j\in S}$, 
	$\{\tilde{a}_{ij}\}_{i,j\in S}$ and  $\{\tilde{m}_j\}_{j\in S}$, 
	respectively, in an independent manner. 
	The max of $Q(\lmb,\cdot)$ occurs exactly at the Cartesian product 
	set of these points.  For (1) and (2), we note that since $\sum_{i=1}^p{\tilde{\eta}_{i}}=1$ 
	and $\sum_{j=1}^p{\tilde{a}_{ij}}=1$ for each $i$, the boundary of 
	their respective optimizations over $\Lambda$ (using Lagrange multipliers for example) occur when at least one term in the 
	sum is 0. However, in this case, the logs approach $-\infty$ so 
	that (1) and (2) cannot be maximized at the boundaries. It follows 
	that any global maxima for these terms must occur at critical points. The reestimation formulae from Algorithm~\ref{algo:baumWelch} 
	give these points 
	(see Theorem~2 in \cite{liporace1982maximum} for derivation). 
	Lastly, term (3) is maximized over the Banach space $H(\gamma)$, so that any extrema are necessarily critical points \cite[Corollary 2.5]{coleman2012calculus} and at least one such point exists by Lemma~\ref{lem:meanreest}. The product of these sets 
	then gives the existence of critical points for $Q(\lambda,\cdot)$. 
	
	Addditionally, since the log function is strictly concave,  terms (1) and (2) are strictly concave in their 
	respective arguments. 
	Due to the strict convexity of squared norms, 
	term (3) is also concave in its argument, becoming strictly 
	concave when $q$ is a norm and ${\pi}_q$ is the identity 
	map on $H(\gamma)$. 
	It follows that every critical point of $Q(\lambda,\cdot)$ 
	is a global	maximum and uniqueness holds due to strict 
	concavity when $q$ is a norm \cite[Theorem 7.4, Proposition 7.3, Proposition 7.4]{coleman2012calculus}.
\end{proof}

\begin{theorem} 
  \label{mono}
  A point $\lmb$ in the parameter space is a critical point 
  of $L_{\lmb}(O)$ if and only if it is a fixed point 
  of the reestimation function, i.e. 
  $Q(\lmb, \lmb) = \max_{\tilde{\lmb}} Q(\lmb, \tilde{\lmb})$. 
  Furthermore 
  $$
    Q(\lmb,\tilde{\lmb})> Q(\lmb,\lmb) 
    \Rightarrow  
    L_{\tilde{\lmb}}(O)> L_{\lmb}(O_t).
  $$ 
  Hence increasing $Q(\lmb, \cdot)$ improves the likelihood. 
\end{theorem}

\begin{proof}
	This closely follows the proof in \cite{LIPORACE1982}. 
	First, we note that for any real valued function $\psi$ on a normed space, $\psi(x)\frac{d\log(\psi(x))}{dx}=\frac{d\psi(x)}{dx}$. 
	Suppose $\lmb$ is a critical point of $L_\lmb(O)$. Then,
	\begin{multline*}
	  0=\nabla_{\tilde{\lmb}}{L_{\tilde{\lmb}}(O)}|_{\tilde{\lmb}=\lmb}
	  =\sum_{s\in S^T}\nabla_{\tilde{\lmb}}{L_{\tilde{\lmb}}(O,s)}|_{\tilde{\lmb}=\lmb}\\
	  = \sum_{s\in S^T}L_\lmb(O,s)\nabla_{\tilde{\lmb}}
	    {\log(L_{\tilde{\lmb}}(O,s))}|_{\tilde{\lmb}=\lmb}
	  = \nabla_{\tilde{\lmb}}Q(\lmb,\tilde{\lmb})|_{\tilde{\lmb}=\lmb}
	\end{multline*}
	implying that $\lmb$ is a critical point of $Q(\lmb,\cdot)$, 
	so that by Theorem~\ref{thm:globalmax}, 
	it is a fixed point of the reestimation. Furthermore, 
	as $\log x \leq x-1$ with equality holding if and only if $x=1$,
	\begin{align*}
	  Q(\lmb, \tilde{\lmb}) - Q(\lmb, {\lmb}) &= 
	  \sum_{s\in S^T} L_\lmb(O,s)\log\left\{
	    \frac{L_{\tilde{\lmb}}(O,s)}{L_{\lmb}(O,s)}
	  \right\}\\
	  &\le \sum_{s\in S^T} 
	  L_\lmb(O,s)\left(\frac{L_{\tilde{\lmb}}(O,s)}{L_{\lmb}(O,s)}-1\right)\\
	  &= L_{\tilde{\lmb}}(O)-L_{\lmb}(O).
	\end{align*}
    Thus, $Q(\lmb, \tilde{\lmb}) > Q(\lmb, {\lmb})$ 
    implies $L_{\tilde{\lmb}}(O)> L_{\lmb}(O)$ as 
    inequality in the above equation is an equality 
    if and only if $L_{\tilde{\lmb}}(O,s)=L_{\lmb}(O,s)$ 
    for each $s\in S$ thereby implying 
    $L_{\tilde{\lmb}}(O)=L_{\lmb}(O)$.
\end{proof}

\subsection{Limits for Model Parameters}

{In this section, we show that that the parameter sequence produced 
by the Baum-Welch algorithm satisfies nice convergence properties
where the parameters have a limit point that is a critical point 
of the likelihood.} 
A similar technique to the classical proofs 
in \cite{wu1983convergence} is used.
This requires two theorems, Berge's Maximum Theorem 
and Zangwill's Convergence Theorem, and some 
preliminary definitions and notation.
Our main result is Theorem~\ref{likelihoodconv}
below.

We recall that the parameter space $\Lambda$ for 
${\lmb}=(\{{\eta}_j\}_{j\in S}$, $\{{a}_{ij}\}_{i,j\in S}$,
$\{{m}_j\}_{j\in S})$ is a closed convex subset of 
$\mathbb{R}^p\times \mathbb{R}^{p\times p}\times H(\gamma)^p$, 
which is additionally compact when projected onto 
$\mathbb{R}^p\times \mathbb{R}^{p\times p}$.
Now by Lemma~\ref{lem:meanreest}, when considering the optimization problem for the likelihood, each $m_j$ can be considered to lie in 
$\overline{\text{conv}}(\{O_t\})+Z$ 
where $\overline{\text{conv}}(\{O_t\})\subset H(\gamma)$ is the 
closed convex hull of the observation sequence.  Without loss of generality, we may take 
$
 m_j={
   \sum_{t=1}^T \alpha_t(j)\beta_t(j){\pi}_{q}(O_t)
 }/{
   \sum_{t=1}^T \alpha_t(j)\beta_t(j)
 }
$ in our reestimation formulas by assigning the component in $Z$ 
to be ${0}$. Then, each state mean $m_j$ lies in
$\overline{\text{conv}}(\{O_t\})$, which is compact 
being the closed convex hull of a finite set of points
\cite[Corollary 5.30]{aliprantis} \cite[Theorem 3.20]{RUDIN_FUNC}. 
Therefore, the arg\,max of 
our objective function over $\lmb$ can be restricted to a 
compact convex subset $\mathcal{C}\subset \Lambda$. We note that in the 
case $q$ is a norm, $Z=\{0\}$
and the reestimation 
formula {automatically restricts to $\mathcal{C}$ 
(see Theorem \ref{thm:globalmax}).}

Given metric spaces $X$, $Y$, a function $f: X \to P(Y)$, 
where $P(Y)$ is the power set of $Y$, is called a \textit{correspondence}
on $X$. Such a map is said to be closed or \textit{upper hemicontinuous} 
if given sequences $\{x_n\}\subset X$ and $\{y_n\}\subset Y$ 
such that $y_n\in f(x_n)$ for each $n\in\natural$, 
$x_n\to x$ and $y_n\to y$ imply $y\in f(x)$. 
On the other hand, it is 
said to be \textit{lower hemicontinuous} if for any sequence 
$\{x_n\}\subset X$, $x_n\to x$ implies that for each 
$y\in f(x)$, there exists a subnet $x_{k_n}$ of $\{x_n\}$ 
and a sequence $\{y_n\}\subset Y$ with $y_n\subset f(x_{k_n})$ 
such that $y_n \to y$. A correspondence that is both upper 
and lower hemicontinuous is said to be continuous. 

\begin{theorem}[Berge's Maximum Theorem] 
  \cite[Theorem 17.31]{aliprantis} \label{berge}
  Let $X$ and $Y$ be Hausdorff topological spaces and 
  let $\phi:X\to P(Y)$ be a continuous correspondence 
  such that $\phi(x)$ is non-empty and compact for all 
  $x\in X$. Additionally suppose $f: X\times Y \to \mathbb{R}$ 
  is continuous. Then the correspondence $\mu: X\to P(Y)$ 
  given by $\mu(x)=\argmax{y\in \phi(x)} f(x,y)$ has non-empty 
  compact values and is upper hemicontinuous. 
\end{theorem}

Let $X=\Lambda$, $Y=\Lambda$ in the above theorem, and let 
$\phi:\Lambda\to P(\Lambda)$ be given by the constant map 
$\phi(x)=\mathcal{C}$ where we recall that $\mathcal{C}$ 
is non-empty, compact, and convex. This is continuous. 
By choosing 
$f:\Lambda\times \Lambda\to \mathbb{R}$ with $f(x,y)= Q(x, y)$,
the map $\lmb\mapsto\mu(\lmb)=\argmax{y\in \mathcal{C}} \ Q(\lmb,y)$ 
is upper hemicontinuous by Theorem~\ref{berge}.

\begin{theorem}[Zangwill's Convergence Theorem] \cite[page 91]{zangwill1969nonlinear}
 Suppose $M$ is a correspondence $M: X\to P(X)$ that 
 generates a sequences $\{x_n\}$ with $x_{n+1}\in M(x_n)$
 that is initiated with $x_0\in X$. 
 Suppose a ``solution set'' $\Gamma\subset X$ is given and 
 suppose that
	\begin{enumerate}
		\item $\{x_n\}\subset S$ for a compact set $S\subset X$
		\item There is continuous function $f$ on $X$ satisfying \begin{enumerate}
			\item if $x\in \Gamma$, then $f(x)\leq f(y)$ for all $y\in M(x)$
			\item if $x\not\in \Gamma$, then $f(x)< f(y)$ for all $y\in M(x)$
			\end{enumerate}
		\item $M$ is upper hemicontinuous on $X\setminus\Gamma$
    \end{enumerate}
Then, every limit point of $\{x_n\}$ lies in $\Gamma$.
\end{theorem}

\begin{theorem} \label{likelihoodconv}
  Suppose $\{\lmb_n\}$ is the sequence generated 
  by Algorithm~\ref{algo:baumWelch} with 
  $M: \Lambda \to P(\Lambda)$ defined as
  $\lmb\mapsto \argmax{\tilde{\lmb} \in C} \ Q(\lmb, \tilde{\lmb})$ 
  and $\lmb_{n+1}\in M(\lmb_n)$. Then, all the limit points of 
  $\{\lmb_n\}$ are critical points of the likelihood, achieving the same likelihood, and the sequence 
  $L_{\lmb_n}(O)$ converges to this value.  Furthermore, at least one such point exists.
\end{theorem}

\begin{proof}
  Let us denote $\mathcal{M}=\argmax{\lmb\in \Lambda} \, L_\lmb(O)$, 
 the set of critical points of the likelihood function of $L(\lambda)=L_\lambda(O)$ on $\Lambda$ and consider $M$ as given.
  Note that the points $\{\lmb_n\}$ generated by 
  Algorithm~\ref{algo:baumWelch} satisfy the following conditions
  \begin{enumerate}
	\item  $\{\lmb_n\}\subset \mathcal{C}\subset \Lambda$ 
	  where $\mathcal{C}$ is compact and convex;
	\item $L: \Lambda \to \real$, $\lmb\mapsto L_\lmb(O)$ 
	  is continuous and by Theorem~\ref{mono}
	\begin{enumerate}
		\item If $\lmb\not\in \mathcal{M}$, then, it is not a 
		fixed point of the reestimation i.e. for all 
		$\tilde{\lmb}\in M(\lmb)$, $Q(\lmb,\tilde{\lmb})>Q(\lmb,\lmb)$ 
		which implies $L(\tilde{\lmb})>L(\lmb)$.
		\item However if $\lmb\in \mathcal{M}$, then by the definition 
		of $M(\lmb)$, $Q(\lmb,\tilde{\lmb})\geq Q(\lmb,\lmb)$ which implies
		$L(\tilde{\lmb})\geq L(\lmb)$.
	\end{enumerate}
    \item Lastly, take $\phi: \Lambda\to P(\Lambda)$ to be the 
    constant correspondence given by $\lmb\mapsto \mathcal{C}$. 
    It is easily checked that this is hemicontinuous and clear 
    that $\phi(\lmb)$ is compact and non-empty for each 
    $\lmb\in \Lambda$. Further, since 
    $Q:\Lambda\times \Lambda \to \real$, 
    $(\lmb, \tilde{\lmb})\mapsto Q(\lmb,\tilde{\lmb})$ 
    is continuous, by Berge's maximum theorem, $M$ is 
    upper hemicontinuous on all of $\Lambda$, whence on 
    $\Lambda\setminus \Gamma$. 
  \end{enumerate}


  As all of the requirements for the Zangwill's Convergence Theorem 
  are satisfied, every limit point of $\{\lmb_n\}$ lies in $\mathcal{M}$
  so that the first part of the theorem holds. 
  Note that by Theorem~\ref{mono} and the monotone convergence theorem
  $\{L(\lmb_n)\}$ converges to $\sup_{n}L(\lmb_n)$. 
  If $\lmb$ is any limit point of $\lmb_n$, then by the uniqueness 
  of the limit, $L(\lmb)=\sup_{n}L(\lmb_n)$ whereby all such limit points give the same likelihood. That such a limit point always exists follows from 
  $\mathcal{C}$ being compact. 
\end{proof}

We conclude this section by noting that if a stronger version of 
Theorem~\ref{mono} holds where every fixed point of the reestimation 
is a local maximum of $L_\lmb(O)$, 
then Theorem~\ref{likelihoodconv} can be extended to the 
likelihood function converging to a local maximum. 
In the special case where either the likelihood or 
the log likelihood functions are concave, 
convergence to the global maximum follows 
due to the properties of concave functions. 

\subsection{Limits for Mixture Distributions}

The previous theorems show that the parameter 
sequence $\{\lmb_i\}_{i=1}^\infty$ produced 
by the Algorithm~\ref{algo:baumWelch} has a unique
limit point when $q$ is a norm.
In this section, we show that this quickly implies
the existence of a weak limit point for the sequence 
of corresponding measures $\{\mu_i\}_{i=1}^\infty$ where 
$$
  \mu_i = 
  \textstyle\sum_{s\in {S}^T} \omega_s^{(i)} \gamma^{(i)}({s})
$$
for weights $\omega_s^{(i)} \propto \prod_{t=1}^T
   \alpha_t^{(i)}(s_t)\beta_t^{(i)}(s_t)$.
These $\mu_i$ are finite mixtures of Gaussian measures.

Let $\lmb_0\in\Lambda$ be the set of initial THMM parameters
and $\gamma_0$ be a centred Gaussian measure on 
a locally convex space $X$.
Each iteration of the Baum-Welch algorithm produces
updated parameters $\lmb_i$, $i=1,\ldots,\infty$,
which consist of a $p$-long vector of initial state 
probabilities $\eta^{(i)}$, a $p\times p$ Markov
transition matrix $A^{(i)}$, and means 
$m_j^{(i)}\in H(\gamma_0)$ for $j=1,\ldots,p$.
Conditioned on a fixed state sequence 
${s}\in {S^T}=\{1,\ldots,p\}^T$ 
at iteration $i$, we have a Gaussian measure 
$
 \gamma^{(i)}({s})=
 \bigotimes_{t=1}^T \gamma^{(i)}_t({s}_t)
$
on the product space $X^{\otimes T}$
where 
$
  \gamma^{(i)}_t({s}_t) = 
  \gamma_0( \cdot - m_{{s}_t}^{(i)})
$
is the Gaussian 
measure $\gamma_0$ on $X$ shifted by 
$m_{{s}_t}^{(i)}$.
Thus, the Baum-Welch algorithm produces a 
sequence of measures 
$\{\mu_i\}_{i=1}^\infty$ on $X^{\otimes T}$, 
which are mixtures of $\abs{{S}^T}=p^T$ 
Gaussian measures.  

\begin{corollary}
  \label{thm:weakConv}
  The sequence of measures $\mu_i$ has a weak limit point.
\end{corollary}
\begin{proof}
  The reestimated means at algorithm iterate $i$ and state $j$,
  $m_{j}^{(i)}$,
  lie in the compact convex hull
  $K=\overline{\text{conv}}\{O_t\}_{t=1}^T$.
  It follows that the sequence 
  $\{m^{(i)}\}_{i\in\natural}$ 
  lying in  
  $K^p$ has a limit point $m = (m_1, \ldots ,m_p)$.  
  Similarly, the Markovian parameters $\eta^{(i)}$ and
  $A^{(i)}$ lie in a compact set and thus have a 
  limit point.
  Let $\omega_s^{(i)}$ be weights corresponding to each state sequence
  $s\in S^T$ where 
  $\omega_s^{(i)} \propto \prod_{t=1}^T
   \alpha_t^{(i)}(s_t)\beta_t^{(i)}(s_t)$
  such that $\sum_{s\in {S}^T}\omega_s^{(i)} = 1$.
  The mapping 
  $$
    \lmb_i = (\eta^{(i)},A^{(i)},m^{(i)}) \rightarrow
    \mu_i = \textstyle\sum_{s\in {S}^T} \omega_s^{(i)} \gamma^{(i)}({s})
  $$
  is continuous {when the latter probability space on $X^{\otimes T}$ 
  is equipped with the weak topology}.  
  Thus, { given a weak limit point $\lmb$ of} $\{\lmb_i\}_{i=1}^\infty$,
  {$\mu$, the corresponding image of the point 
  under the above map,} is a weak limit point for $\mu_i$.
\end{proof}

\section{Simulated Data Analysis}
\label{sec:data}

In the following sections, we test our THMM algorithm on  
a variety of simulated datasets from three different 
settings: Brownian motion with linear drift, 
the Ornstein-Uhlenbeck process, and 
fractional Brownian motion with linear drift.
To evaluate 
its performance accuracy, we use the adjusted Rand index (ARI)
as our performance metric.  The ARI is a popular method of
measuring the agreement between two sets of labels and
is computed in R via the \texttt{adjustedRandIndex()} 
function in the \texttt{mclust} package \citep{MCLUST}.
An ARI value of 1 indicates a perfect match whereas 
an ARI value of 0 indicates random guessing.

For each of the following simulations, it is possible to 
concoct an algorithm specifically designed to perform 
well on that specific dataset; e.g. via feature selection.  
The true power of the 
THMM approach is that it is generally applicable and 
adaptable to all of these settings of interest as
well as others not considered in this work.

\subsection{Brownian Motion with Drift}
\label{sec:BMWD}

For a simple setting to test the THMM algorithm,
we simulate a sequence of $T=200$ Brownian sample
paths with 5 different states corresponding to 
different drift parameters.  For the ``low separation''
simulation, the drift parameters are $-4,-2,0,2,4$.
For the ``medium separation''
simulation, the drift parameters are $-8,-4,0,4,8$.
In both cases the Markov transition matrix is 
$$
  A = \begin{pmatrix}
    0.64 & 0.09 & 0.09& 0.09& 0.09\\
    0.09 & 0.64 & 0.09& 0.09& 0.09\\
    0.09 & 0.09 & 0.64 & 0.09& 0.09\\
    0.09 & 0.09& 0.09& 0.64 & 0.09\\
    0.09 & 0.09& 0.09& 0.09&0.64  \\
  \end{pmatrix}
$$
This data is displayed in Figure~\ref{fig:bmwdSim}.

The results of running the Baum-Welch and 
Viterbi algorithms making 
use of the Onsager-Machlup functional for Brownian 
motion with drift, see Section~\ref{sec:OMWiener}, 
are displayed in Table~\ref{tab:bmwdSim}.
The ARI is 0.457 for the low separation setting
and 0.842 for the medium separation setting.
The estimated drift parameters are 
$-3.40,-0.07,1.50,2.28,3.54$ for low separation
and $-8.00,-3.83,0.85,4.02,7.63$ for medium.
The estimated transition matrices are 
$$
  \hat{A}_\mathrm{low} = \begin{pmatrix}
    0.63 &0.04 &0.12 &0.14 &0.08\\
    0.24 &0.53 &0.07 &0.05 &0.12\\
    0.20 &0.07 &0.58 &0.12 &0.04\\
    0.04 &0.13 &0.37 &0.22 &0.25\\
    0.05 &0.31 &0.28 &0.25 &0.11\\
  \end{pmatrix}
  \text{ and }
  \hat{A}_\mathrm{med} = \begin{pmatrix}
    0.69 &0.22 &0.02 &0.06 &0.00\\
    0.06 &0.61 &0.11 &0.11 &0.11\\
    0.05 &0.15 &0.53 &0.10 &0.17\\
    0.08 &0.09 &0.08 &0.68 &0.07\\
    0.13 &0.08 &0.12 &0.04 &0.64\\
  \end{pmatrix}.
$$

\begin{figure}
    \centering
    \includegraphics[width=0.45\textwidth]{\PICDIR/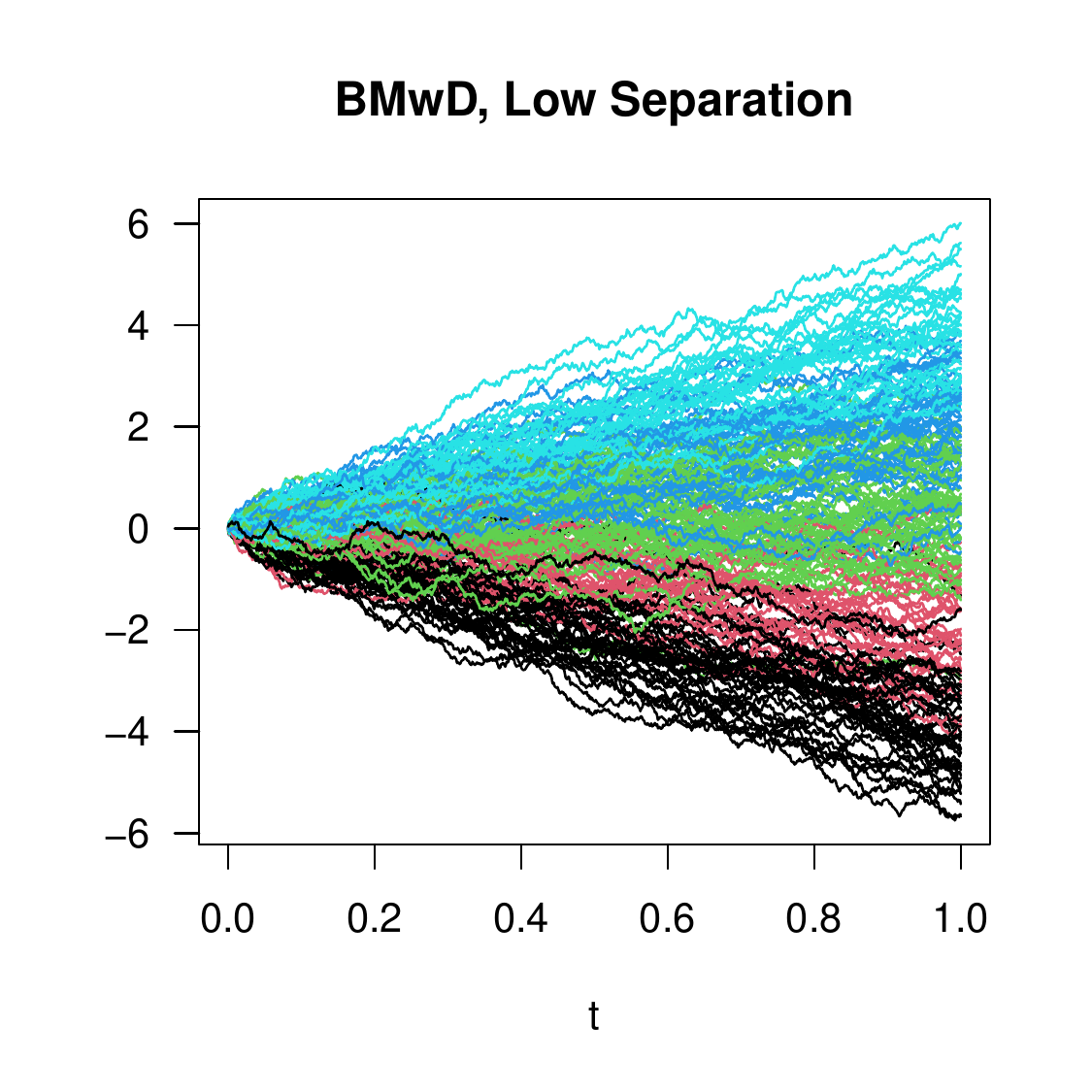}
    \includegraphics[width=0.45\textwidth]{\PICDIR/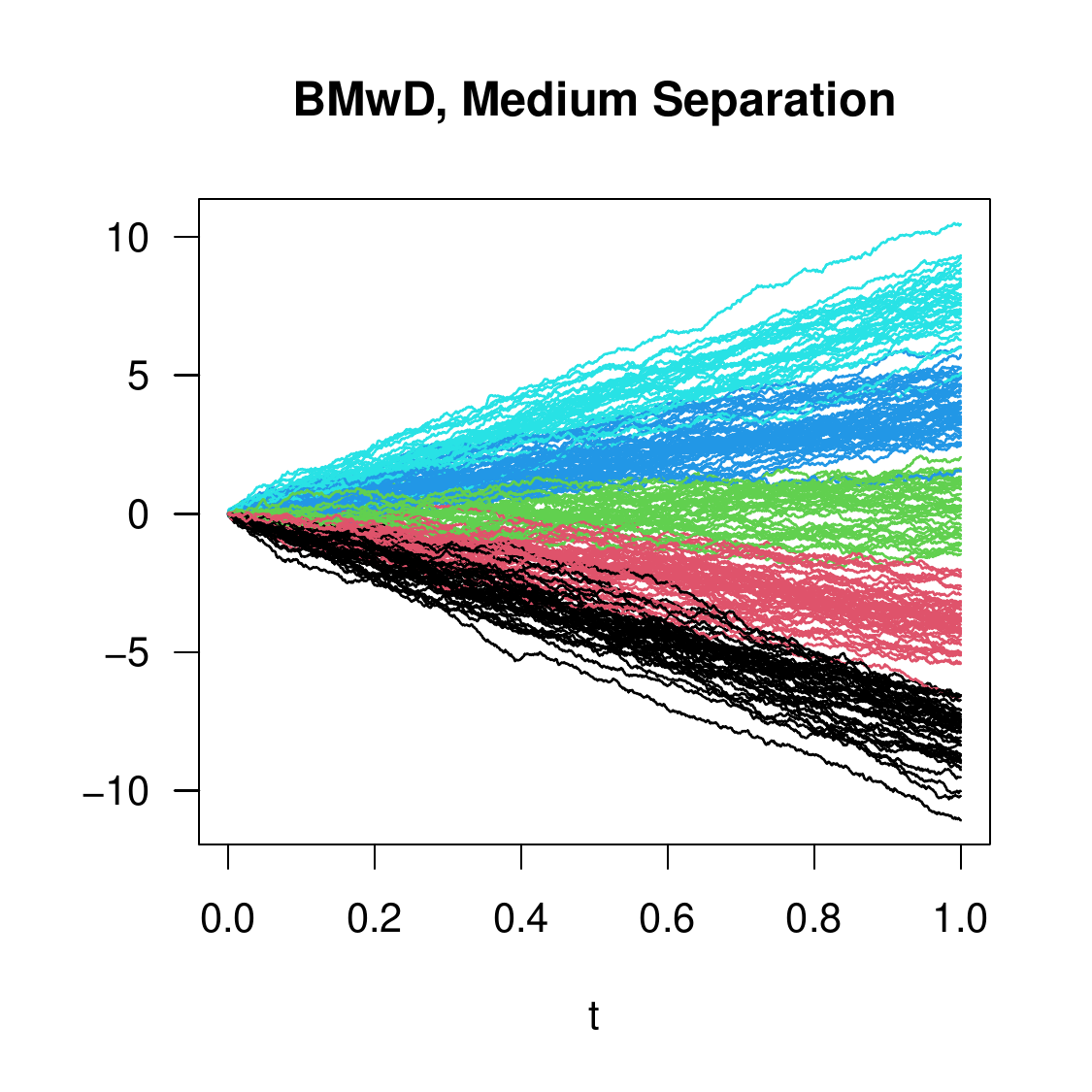}
    \caption{
      Simulated Brownian motion with five states with drift 
      parameters $(-4,-2,0,2,4)$ on the left and 
      $(-8,-4,0,4,8)$ on the right.
    }
    \label{fig:bmwdSim}
\end{figure}

\begin{table}
    \centering
    \begin{tabular}{rrrrrrr}
    \hline\hline
    &\multicolumn{6}{c}{\bf BMWD Low Separation}\\
    & \multicolumn{5}{c}{\bf True States} & 
      \multicolumn{1}{c}{Est} \\
        &A& B& C& D& E& \multicolumn{1}{c}{drift} \\
        \hline
      a &41 &20 & . & . & . & -3.40\\
      b &1  &8  &42 &13 &.  & -0.07\\
      c &.  &1  &7  &15 &.  &  1.50 \\
      d &.  &.  &1  &10 &30 &  2.28 \\
      e &.  &.  &.  &8  &3  &  3.54 \\
   drift&-4 & -2 & 0 & 2 & 4 \\
   \hline\hline
    \end{tabular}
    \begin{tabular}{rrrrrrr}
    \hline\hline
    &\multicolumn{6}{c}{\bf BMWD Medium Separation}\\
    & \multicolumn{5}{c}{\bf True States} & 
      \multicolumn{1}{c}{Est} \\
        &A& B& C& D& E& \multicolumn{1}{c}{drift} \\
        \hline
      a &41 & . & . & . & . & -8.00\\
      b &. & 46 & 3 & . & .  & -3.83\\
      c &. & . & 31 & 8 & .  &  0.85 \\
      d &. & . & 1 & 39 & 1 &  4.02 \\
      e &. & . & . & 1 & 29  &  7.63 \\
   drift&-8 & -4 & 0 & 4 & 8 \\
   \hline\hline
    \end{tabular}
    \caption{
      Confusion matrices showing the alignment of
      true states (A-E) with estimated states (a-e).
      The left is low separation with an ARI of 0.457.
      The right is medium separation with an ARI of 0.842.
    }
    \label{tab:bmwdSim}
\end{table}

\subsection{Ornstein-Uhlenbeck Process}
\label{sec:OUP}

The one-dimensional Ornstein-Uhlenbeck process has 
the form 
$$
  dY_\tau = c(\mu - Y_\tau) d\tau + dW_\tau
$$
with two parameters.
The mean parameter $\mu$ is where the process 
tends to in the long run, and the concentration 
parameter $c$ determines how tightly the process
fluctuates around its mean.  However, this form 
is numerically challenging to optimize.  Instead, 
the THMM estimates the transformed variables
$b_0 = c\mu$ and $b_1 = c$. This is implemented 
in \texttt{R} via the \texttt{optim} function
with the \texttt{L-BFGS-B} method.  The function
to minimize is 
$$
  u({b_0,b_1}) =  \frac{\sum_{t=1}^T
  \alpha_t(j)\beta_t(j)\left\{
    \frac{1}{2}
    \int_{0}^1 \{
      \dot{O}_{t,\tau} - (b_0 - b_1 O_{t,\tau})
    \}^2d\tau - \frac{b_1}{2}
  \right\} }{
    \sum_{t=1}^T
    \alpha_t(j)\beta_t(j)
  }
$$
with derivatives 
\begin{align*}
  \frac{\partial u}{\partial b_0} &= 
  -\frac{\sum_{t=1}^T
  \alpha_t(j)\beta_t(j)\left\{
    \int_{0}^1 \{
      \dot{O}_{t,\tau} - (b_0 - b_1 O_{t,\tau})
    \}d\tau 
  \right\} }{
    \sum_{t=1}^T
    \alpha_t(j)\beta_t(j)
  }\\
  \frac{\partial u}{\partial b_1} &= 
  \frac{\sum_{t=1}^T
  \alpha_t(j)\beta_t(j)\left\{
    \int_{0}^1 \{
      \dot{O}_{t,\tau} - (b_0 - b_1 O_{t,\tau})
    \}{O}_{t,\tau}d\tau - \frac{1}{2} 
  \right\} }{
    \sum_{t=1}^T
    \alpha_t(j)\beta_t(j)
  }
\end{align*}
where we divide everything by 
$\sum_{t=1}^T\alpha_t(j)\beta_t(j)$ for
numerical stability reasons.

In this simulation, five states were once again
used to generate data with means $\mu = (-2,0,4,2,1)$ and
$c = (4,4,8,2,20)$.  The transition matrix
is the same as used above for Brownian motion
with drift and $T=200$ again.  The OU sample paths
coloured by their state are displayed in 
Figure~\ref{fig:oupSim}.

Table~\ref{tab:oupSim} displays the results 
of the Baum-Welch and Viterbi algorithms on
this simulated dataset. The estimated and true 
mean parameters are
$$
  \begin{pmatrix}
    \hat{\mu}\\\mu
  \end{pmatrix} =
  \begin{pmatrix}
    -1.92&-0.30&3.32&2.06&0.80\\
    -2&0&4&2&1
  \end{pmatrix}
$$
showing good parameter recovery.  In contrast, 
the estimated and true concentration parameters
do not align as well
$$
  \begin{pmatrix}
    \hat{c}\\c
  \end{pmatrix} =
  \begin{pmatrix}
    4.75&8.26&4.94&4.58&3.80\\
    4&4&8&2&20
  \end{pmatrix}.
$$
Regardless, the achieved ARI is 0.678.  The estimated 
transition matrix is similar to the true transition
matrix:
$$
\hat{A}_\mathrm{OU} = \begin{pmatrix}
 0.56 & 0.06 & 0.16 & 0.12 & 0.09 \\
 0.21 & 0.58 & 0.16 & 0.00 & 0.05 \\
 0.06 & 0.12 & 0.59 & 0.08 & 0.16 \\
 0.00 & 0.05 & 0.39 & 0.42 & 0.14 \\
 0.08 & 0.03 & 0.21 & 0.04 & 0.65 
\end{pmatrix}.
$$

\begin{figure}
    \centering
    \includegraphics[width=0.65\textwidth]{\PICDIR/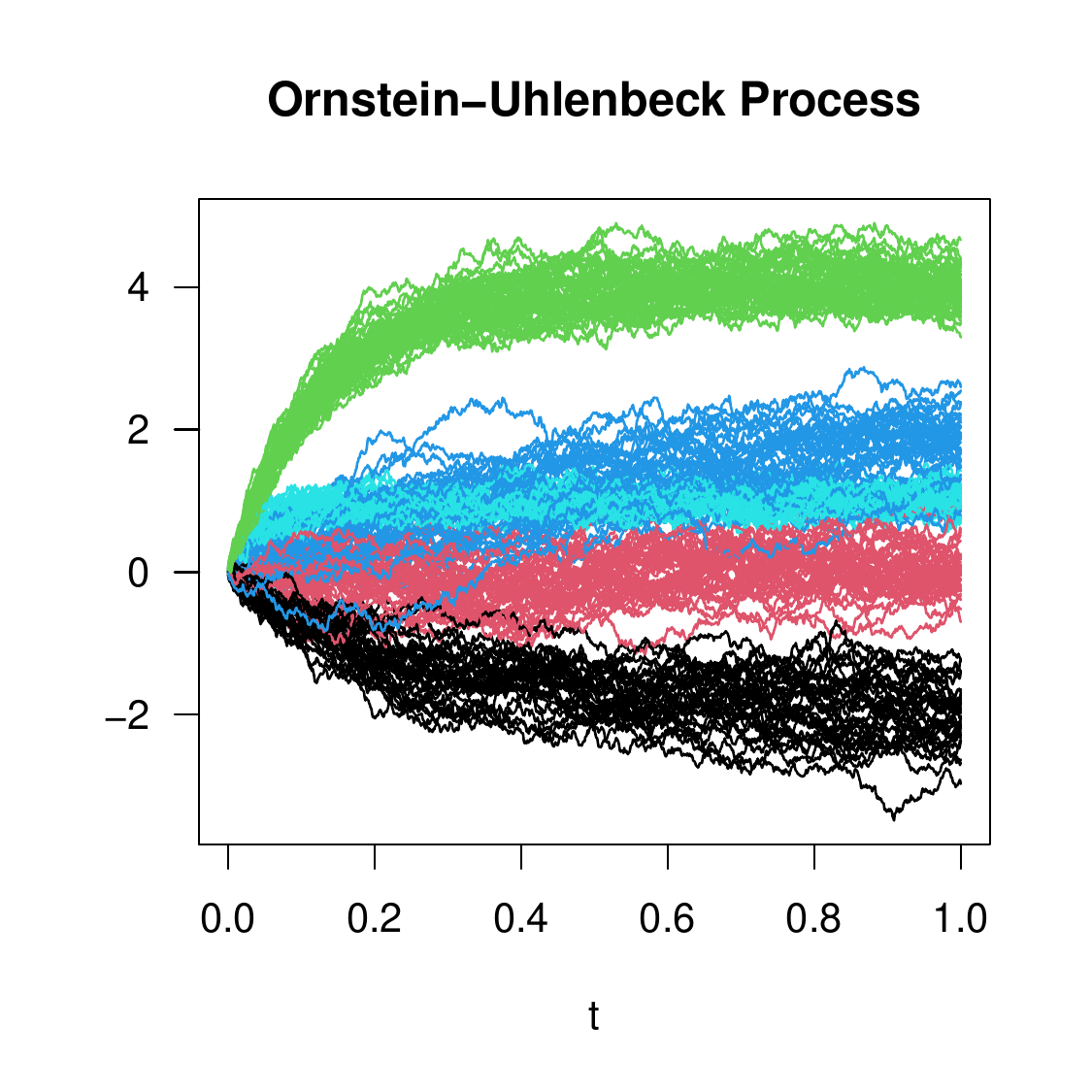}
    \caption{
      Simulated Ornstein-Uhlenbeck Processes with 5 
      states, mean parameters $\mu = (-2,0,4,2,1)$ and
      concentration parameters $c = (4,4,8,2,20)$.
    }
    \label{fig:oupSim}
\end{figure}

\begin{table}
    \centering
    \begin{tabular}{rrrrrrrr}
    \hline\hline
    &\multicolumn{6}{c}{\bf OU Process}\\
    & \multicolumn{5}{c}{\bf True States} & 
      \multicolumn{2}{c}{Estimated} \\
        &A& B& C& D& E& \multicolumn{1}{c}{mean}
        &\multicolumn{1}{c}{conc}\\
        \hline
      a &32& .&  .&  .&  .& -1.92& 4.75\\
      b &. &28&  .&  .&  .& -0.30& 8.26\\
      c &. & .& 51&  .&  .&  3.32& 4.94\\
      d &. & .&  .& 21&  .&  2.06& 4.58\\
      e &. &12&  .& 21& 35&  0.80& 3.80\\
   mean &  -2&0&4&2&1  \\
   conc &  4&4&8&2&20  \\
   \hline\hline
    \end{tabular}
    \caption{
      Confusion matrix showing the alignment of
      true states (A-E) with estimated states (a-e)
      for the OU process data.
      The ARI is 0.678.
    }
    \label{tab:oupSim}
\end{table}

\subsection{Fractional Brownian Motion}
\label{sec:FracBM}

\begin{figure}
    \centering
    \includegraphics[width=0.45\textwidth]{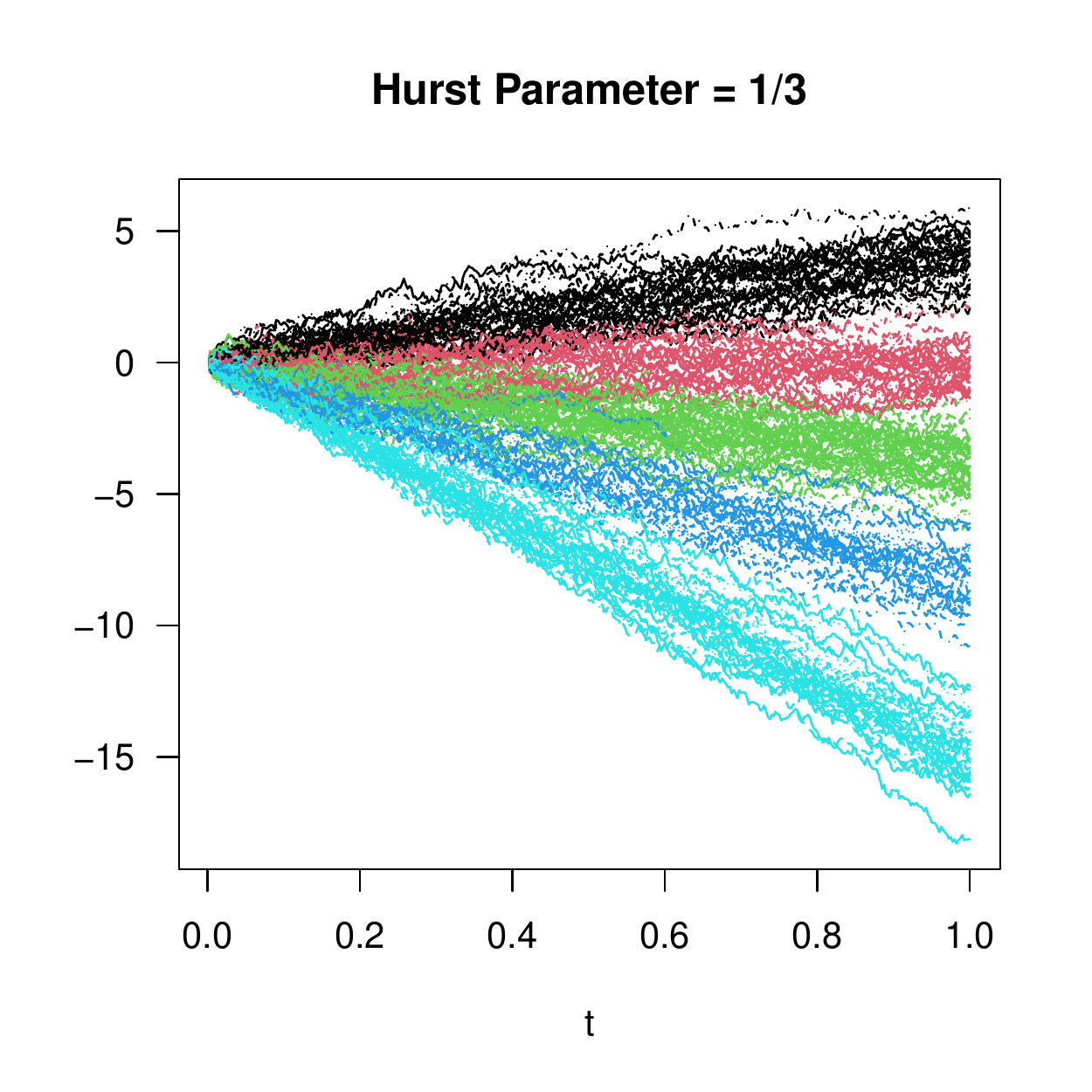}
    \includegraphics[width=0.45\textwidth]{\PICDIR/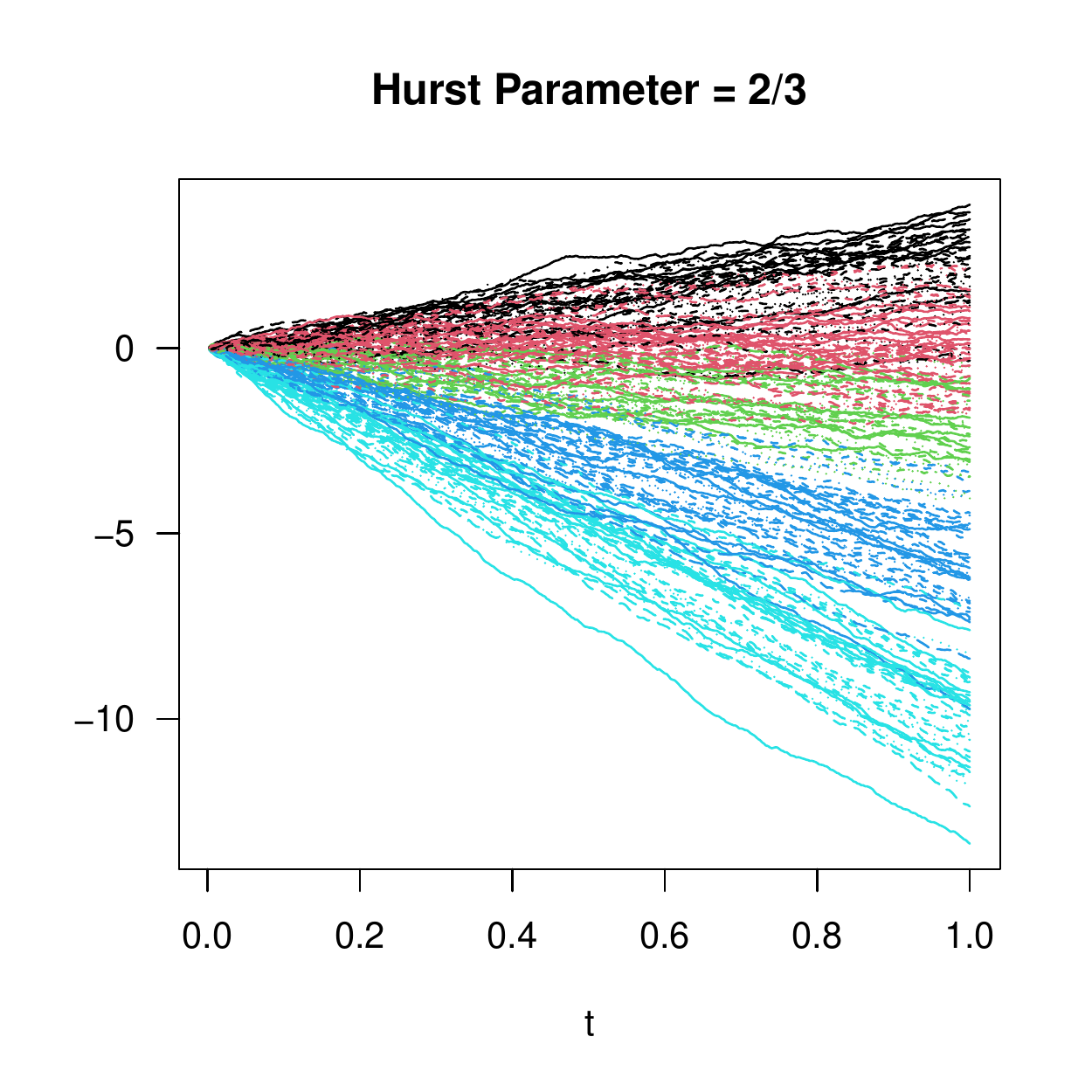}
    \caption{
      \label{fig:fracBMExamples}
      Example sample paths of fractional Brownian motion
      with Hurst parameter of 1/3 (left) and 2/3 (right).
    }
\end{figure}

To test the THMM algorithm applied to fractional 
Brownian motion, we first simulate 200 sample paths
with a Hurst parameter of $\nu=2/3$, which gives a Gaussian
process with positively correlated increments, i.e. it is 
\textit{smoother} than the standard Wiener process. 
Simulation 
was achieved by first simulating Gaussian white noise
and transforming it based on the covariance function
$$
  \cov{Y_{\tau_1}}{Y_{\tau_2}} = \frac{1}{2}\left(
    \tau_1^{2\nu} + \tau_2^{2\nu} - 
    \abs{\tau_1-\tau_2}^{2\nu}
  \right).
$$
This data is displayed in Figure~\ref{fig:fracBMExamples}.
The same transition matrix $A$ was used and the 5-long vector
of drift terms is $c = (2,0,-2,-6,-10)$.  In this work,
we treat the Hurst parameter as a tunable input to the 
THMM algorithm rather than a parameter to be estimated 
from the data.  Hence, we run the THMM algorithm for 
$\nu = 0.5, 2/3, 0.9$ to compare performance.

Table~\ref{tab:fracBMSmooth} shows the results of this
simulation.  The estimated drift vector is comparable 
between all three runs.  Choosing the correct Hurst 
parameter $\nu=2/3$ gave the highest ARI value.  However,
the other two runs are not far behind.

\begin{table}
    \centering
    \begin{tabular}{l|rrrrr|r}
    \hline
      \multicolumn{1}{c}{}
      & \multicolumn{5}{c}{Drift Vector} & ARI\\
    \hline
    Truth &  -10.\phantom{00} & -6.\phantom{00} & 
    -2.\phantom{00} &0.\phantom{00} &2.\phantom{00} & \\
    $\nu=0.5$  & -10.02& -5.72 & -1.65 &  0.37 &  2.05 & 0.557\\
    $\nu=2/3$  & -9.48 & -5.24 & -1.48 &  0.56 &  1.99 & {\bf 0.593}\\
    $\nu=0.9$  & -10.28& -5.83 & -1.51 &  0.87 &  2.06 & 0.589\\
      \hline
    \end{tabular}
    \caption{
      \label{tab:fracBMSmooth}
      A comparison of the THMM algorithm run on fractional
      Brownian motion with Hurst parameter of 2/3 for 
      different choices of $\nu$ as an input to the
      algorithm.
    }
\end{table}

We repeat the same experiment but for fractional Brownian motion
with a Hurst parameter of $1/3$, which gives negatively
correlated increments and \textit{rougher} looking paths.  The
true drift coefficients are set to be 
$(4,0,-4,-8,-15)$ as the paths are harder to distinguish 
than in the previous simulation.
The results from Table~\ref{tab:fracBMRough} show that 
setting the Hurst parameter in the THMM algorithm to 
the $\nu=1/3$ or $1/4$ gave the better ARI values than 
simply assuming we have standard Brownian motion.

\begin{table}
    \centering
    \begin{tabular}{l|rrrrr|r}
    \hline
      \multicolumn{1}{c}{}
      & \multicolumn{5}{c}{Drift Vector} & ARI\\
    \hline
    Truth &  -15.\phantom{00} & -8.\phantom{00} & 
    -4.\phantom{00} &0.\phantom{00} &4.\phantom{00} & \\
    $\nu=1/2$  & -14.31 & -6.82 & -3.91 & -1.00 & 3.24 & 0.562\\
    $\nu=1/3$  & -16.44 & -9.71 & -4.82 & -1.23 & 3.56 & {\bf 0.622}\\
    $\nu=1/4$  & -17.56 & -9.75 & -5.04 & -1.25 & 3.83 & {\bf 0.622}\\
      \hline
    \end{tabular}
    \caption{
      \label{tab:fracBMRough}
      A comparison of the THMM algorithm run on fractional
      Brownian motion with Hurst parameter of 1/3 for 
      different choices of $\nu$ as an input to the
      algorithm.
    }
\end{table}


\section{Pediatric Obstructive Sleep Apnea Data}
\label{sec:pedOSA}

\subsection{Data Overview}

Obstructive sleep apnea is a chronic condition characterized 
by frequent episodes of upper airway collapse during sleep. 
The gold standard for diagnosis of obstructive sleep apnea 
in children is by overnight polysomnography in a hospital 
or sleep clinic. Polysomnography provides multi-channel 
time series including an electroencephalogram (EEG),
electrocardiography (ECG), electrooculography (EOG),
and electromyography (EMG).

Even for a single patient, this data is vast
and should eventually be considered jointly 
to label sleep states and identify sleep disorders.
However
to illustrate application of the THMM algorithms
in this work, 
we chose one channel of EEG 
from one patient labeled as CF050
to be used as a proof-of-concept. 
This patient was selected from a group of
seventy four pediatric 
patients with potential obstructive sleep apnea 
in a clinical study, 
Pro00057638, approved by University of Alberta.
The sampling rate for EEG is 512 samples
per second. 
Each signal was split into a sequence of epochs,
i.e. 30 second intervals.
These signals were transformed into power spectral 
densities (PSD) using Welch’s method \citep{welchperiod} 
for each epoch.  The reason for 
dividing time series into 30 second intervals is 
to group them with respect to five sleep stages: 
wake, rapid eye-movement (REM), and non-rapid 
eye-movement (NREM).  The NREM category is further 
divided into three states: NREM1, NREM2, NREM3.  
The sleep stages are labelled per epoch by a sleep 
technician. 
For every sleep stage not labelled as \textit{wake}, 
the patient is considered to be asleep. The spectral 
densities for each epoch may then be used to identify 
possible underlying states such as the sleep stages.
There were 948 epochs in patient CF050.

Markov models and HMMs have a long history of being 
applied to the modelling of sleep states
\citep{HMM-Zung1965ComputerSO,HMM-Yang1973TheUO,HMM-Kemp1986SimulationOH}.
The work of \cite{HMM-penny1998gaussian} considered an
HMM with Gaussian observations for EEG analysis.
This was followed by similar analyses in 
\cite{HMM-Flexerand,HMM-Flexer}.
More recent work includes
\cite{HMM-DKS}, \cite{HMM-pan2012transition}, and
\cite{HMM-Chen}.
Typically, these approaches are based on various 
discretization and feature selection methods whose
output is then fed into the classic HMM.  Our THMM
takes the entire power spectral density curve into account and 
hence obviates the need for such feature selection 
steps.
Note that the following data analysis is meant as a 
proof-of-concept for the proposed methodology; a 
comprehensive analysis of the full multichannel dataset 
is left to future work.

\subsection{Raw Data Analysis}

The EEG PSD curves for patient CF050 was run through the THMM
variant of the Baum-Welch algorithm under Wiener measure
and the $H=W_0^{2,1}[0,1]$ norm.  
The first experiment combined REM and the NREM states
into one category ``asleep'' to contrast with the
``awake'' category.
Figure~\ref{fig:eeg50TwoState} shows the data and
the predicted vs the true state means.  The estimated
mean curves are very similar to the true mean curves.
The estimated Markov transition matrix and the
``true'' transition probabilities are, respectively,
$$
  \hat{A} = \begin{pmatrix}
    0.910 & 0.090 \\
    0.176 & 0.824
  \end{pmatrix}
  ~\text{ and }
  {A} = \begin{pmatrix}
    0.933 & 0.067 \\
    0.025 & 0.975
  \end{pmatrix}
$$
where $A$ was computed by counting the number of 
transitions between states as labelled by the sleep
technician.
It is worth noting that patients typically remain 
asleep or awake for many sequential epochs, i.e. 30-second
time segments.

\begin{figure}
    \centering
    \includegraphics[width=0.45\textwidth]{\PICDIR/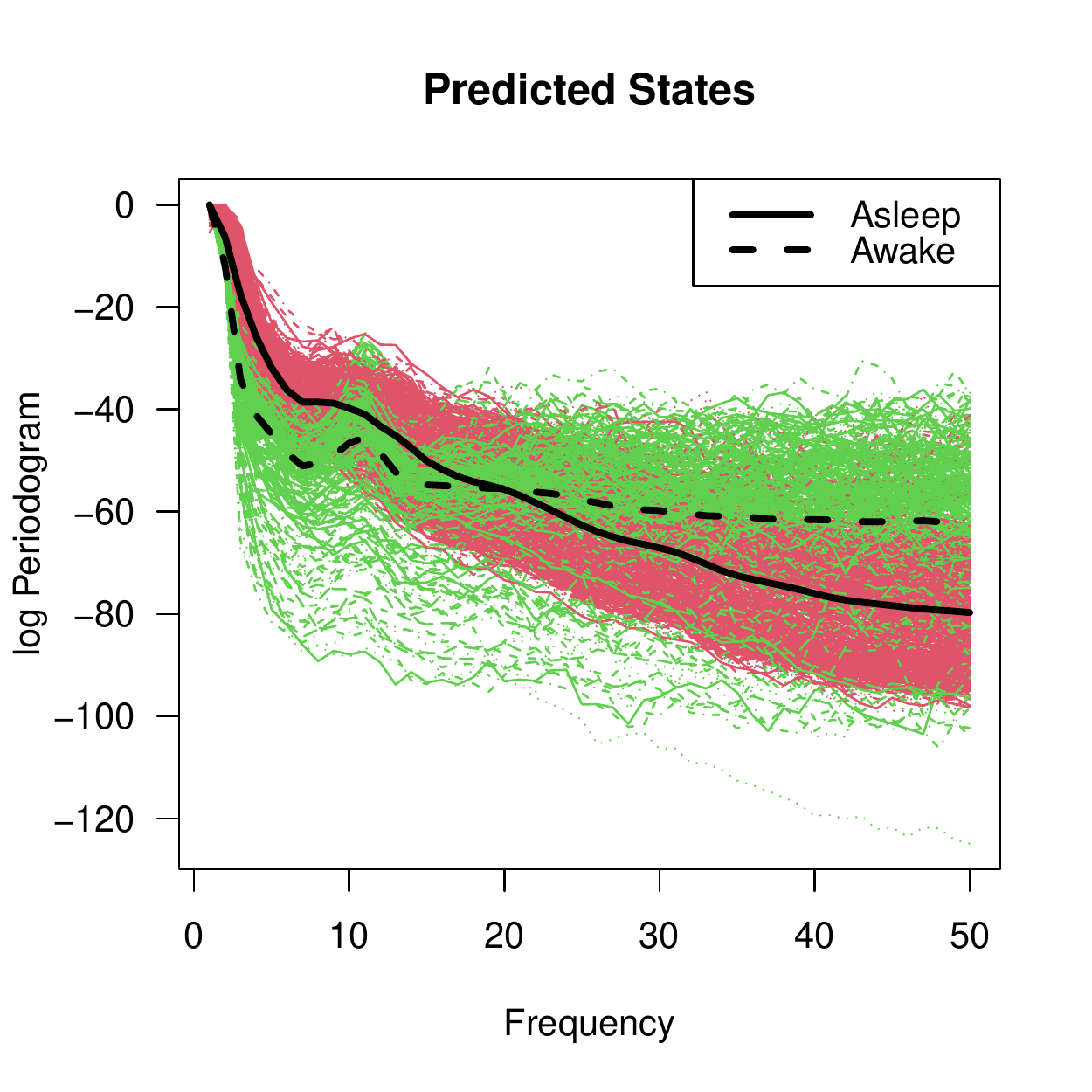}
    \includegraphics[width=0.45\textwidth]{\PICDIR/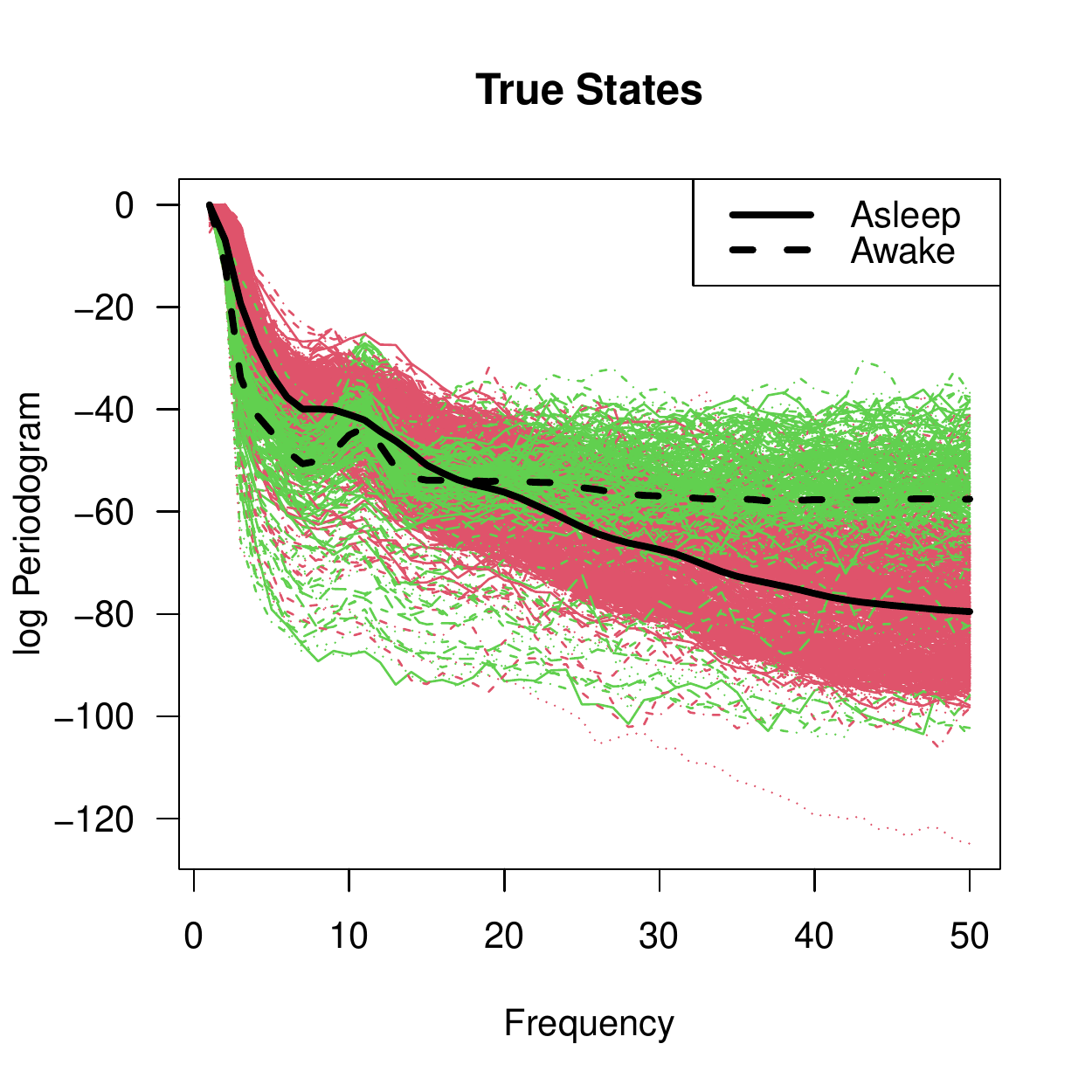}
    \caption{
      \label{fig:eeg50TwoState}
      A comparison of the predicted and true states and means 
      for the EEG PSD dataset over the first 50 frequencies.
    }
\end{figure}

For this experiment, a particularly
poor starting point for the Baum-Welch algorithm
was chosen so that the progress of the algorithm
could be tracked with respect to both the 
likelihood and the ARI.  The left plot in 
Figure~\ref{fig:eeg50ARI} shows
how the likelihood grows over the initial 
iterations only to plateau around iteration 10.
However, it begins to climb to a new plateau after
30-40 iterations.  In unison, the right plot shows
improvement in the ARI from below zero to its 
final value of 0.68.  The final confusion 
matrix is displayed on the left side of 
Table~\ref{tab:confusePOSA}, and a 
comparison of the predicted state sequence 
via the Viterbi algorithm with the true state
sequence is featured in Figure~\ref{fig:osaBestStates}.
Most alternative starting points for the 
algorithm converged
to the same final parameters.  In some cases, the 
ARI rose above 0.680 only to fall back to it.
Hence, it is worth emphasizing that each iteration
of the Baum-Welch algorithm will increase the likelihood
but may result in an increase or a decrease in the ARI.

\begin{table}
    \centering
    \begin{tabular}{rrrcrr}
      \hline\hline 
    & \multicolumn{5}{c}{\bf True States} \\
    & \multicolumn{2}{c}{\bf Raw Data}
    && \multicolumn{2}{c}{\bf Smoothed Data}\\
        Predicted
        &Awake& Asleep &&Awake& Asleep \\
        \hline
       Awake  & 620 &   6 &&  649 &  12\\
       Asleep &  75 & 247 &&   46 & 241\\
   \hline\hline
\end{tabular}
    \caption{
      \label{tab:confusePOSA}
      Confusion matrices for the results of fitting
      a two-state THMM to raw and kernel smoothed 
      EEG PSD curves, left and right, respectively. 
    }
    
\end{table}

\begin{figure}
    \centering
    \includegraphics[width=0.45\textwidth]{\PICDIR/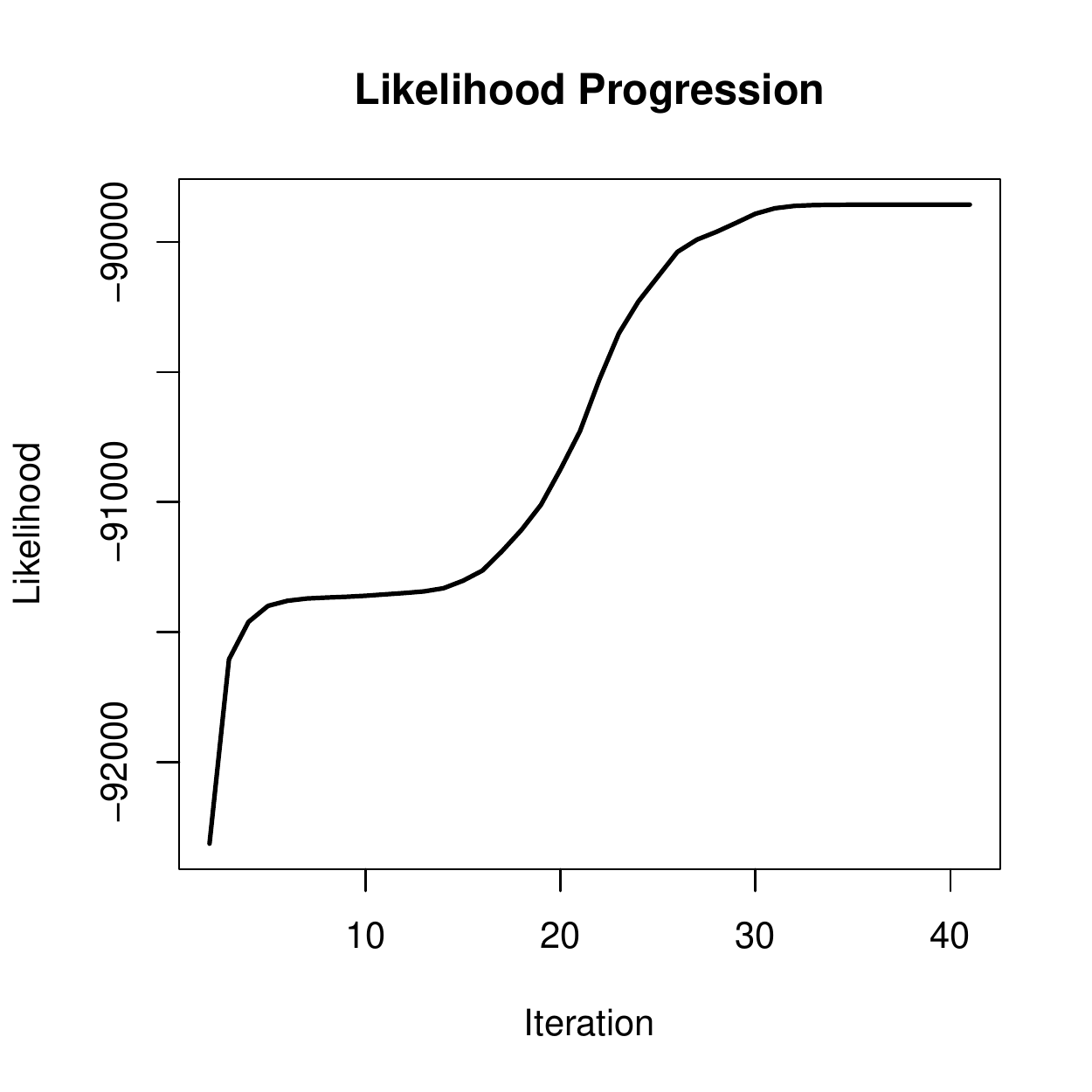}
    \includegraphics[width=0.45\textwidth]{\PICDIR/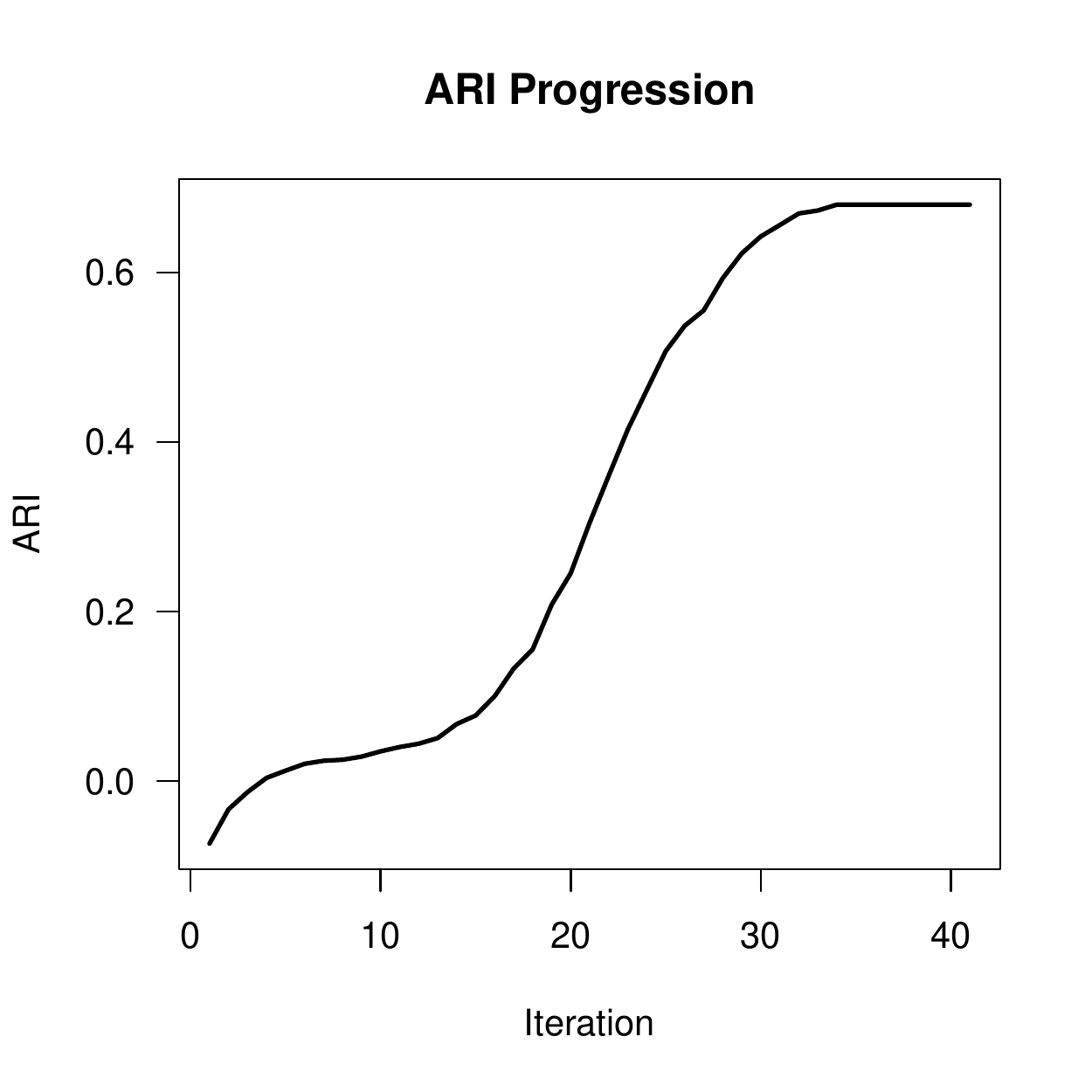}
    \caption{
      \label{fig:eeg50ARI}
      Tracking the increase in likelihood (left) and ARI (right)
      as the THMM algorithm runs.
    }
\end{figure}

\begin{figure}
    \centering
    \includegraphics[width=0.75\textwidth]{\PICDIR/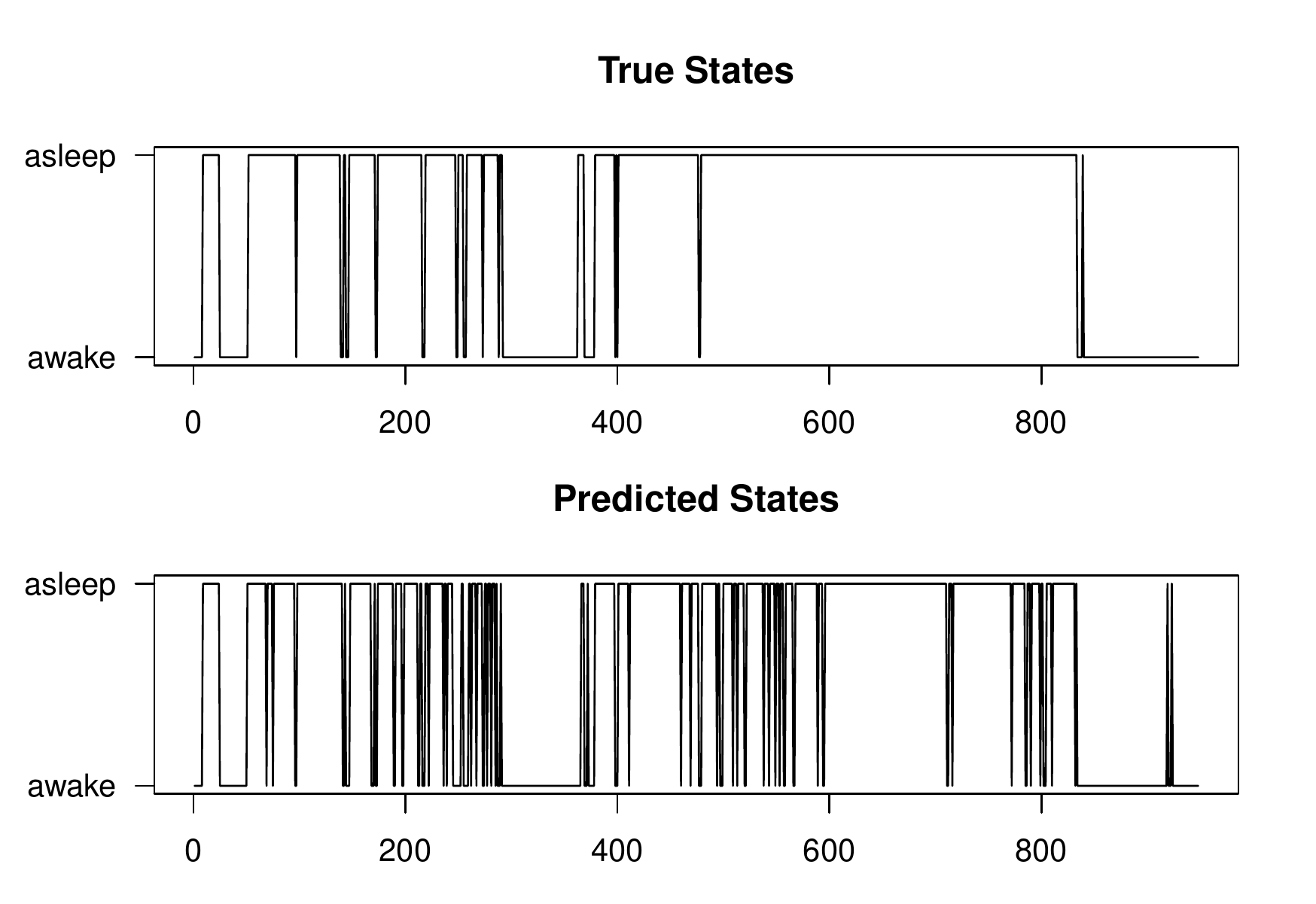}
    \caption{
      \label{fig:osaBestStates}
      Comparison of the predicted state sequence 
      via the THMM and the true state sequence based
      on the raw EEG PSD curves.
    }
\end{figure}

Performing the same analysis for a three-state 
THMM model on the same data yields a final 
ARI of 0.335.  In this case, the algorithm has 
a harder time separating the REM and NREM states, 
which have similar trajectories.  

\subsection{Kernel Smoothed Data Analysis}

A subtle point from the above theory is that
the observations $O_t$
do not remain as elements of the LCTVS $X$, 
but are instead considered within the Cameron-Martin
space $H$.  In the
case of Wiener measure and $W_0^{2,1}[0,1]$,
that implies that we want a smooth analogue 
of $O_t$.  Thus, we can use, for example, a 
Gaussian kernel smoother to create a smoothed 
version of each $O_t$.  
Repeating the above experiment for for a 
smoothed EEG PSD curves
with two states improves the ARI from 0.680 
to 0.762.  The confusion matrix is displayed
on the right side of Table~\ref{tab:confusePOSA}.

In contrast, attempting a three-state THMM 
on the smoothed EEG PSD curves with the 
$W_0^{2,1}[0,1]$ norm resulted in a marginally
smaller ARI of 0.328 when compared to the 
original unsmoothed curves.  However,
using instead the $W_0^{2,2}[0,1]$ norm on the 
smoothed curves marginally 
improved the ARI to 0.368.  Thus, future investigations
into this and the full five-state model will require
further innovation.  This potentially includes 
fitting a THMM to multichannel signals data and
considering other Gaussian measures / Cameron-Martin 
spaces to work within.

\subsection{Comparison with Functional k-means}

We compare the performance of our THMM method
with functional data k-means clustering 
\citep{SANGALLIkma1,SANGALLIkma2}
as implemented in the R package 
\texttt{fdakma} \citep{FDAKMA}.
While this k-means algorithm is agnostic to the 
time ordering of the data, its clustering criteria
is very similar in style to THMM based on the 
Cameron-Martin norm.  Specifically, the \texttt{kma}
function allows for a variety of dissimilarity measures
between pairs of functions $f$ and $g$ including the $L^2[0,1]$
distance $\norm{f-g}_{L^2}$ and the $W_0^{2,1}[0,1]$ distance
$\norm{f'-g'}_{L^2}$ as well as the methods that compute
the Pearson similarity between $f$ and $g$ or $f'$ and $g'$.

We compare our THMM performance to the functional data k-means
algorithm
for both raw and smoothed data and 2 and 3 clusters
by running each algorithm in each experimental setting
for 20 random starts. 
The results are displayed in Table~\ref{tab:kmaCompare}.
For raw (unsmoothed) EEG PSD curves, the THMM achieves the highest ARI
on average over the 20 random starts.  Also, the standard 
deviations are close to or equal to zero showing that both 
the THMM and k-means algorithms consistently find the 
same optima.  In contrast, attempting to split the smoothed EEG
PSD curves into two clusters resulted in large standard deviations 
for the THMM and k-means in $W_0^{2,1}$ and consistently poor
performance for k-means in $L^2$.  The THMM achieved an ARI of
0.762 in 17 out of 20 runs and 0.00 in the other three runs.
The functional k-means achieved 0.768 in 16 out of 20 runs and 
0.00 in the other four.  In contrast, when the number of clusters
is set to three, the THMM consistently achieved an ARI of 0.329
on the smoothed EEG PSD curves.  
The functional k-means mostly returned
an ARI of 0.312, but rose to 0.571 for two of the 20 runs.

In summary, the THMM algorithm outperforms functional data k-means
in $W_0^{2,1}$ for this set of raw unsmoothed EEG PSD curves.  After 
passing the data through a Gaussian kernel smoother, the performance
of the two methods is more comparable.  Running functional data k-means
in $L^2$ seems to perform more poorly than the other methods in all cases.
We also note that the k-means algorithm does not immediately yield 
a measure of how \textit{good} a given clustering is whereas the 
likelihood can be compared between multiple THMM runs to choose
the overall optimal one.
It may also be of future interest to use the functional data k-means
algorithm as a way to initialize the THMM especially for more massive
and complex datasets.

\begin{table}
  \centering
  \begin{tabular}{lrrrr}
    \hline
                        & \multicolumn{2}{c}{Two Clusters} & 
                          \multicolumn{2}{c}{Three Clusters}\\
                method  & Raw & Smooth & Raw & Smooth\\
                        \hline
    THMM $W_0^{2,1}$    & {\bf 0.680} (0.00) &  {\bf 0.641} (0.29) & {\bf 0.347} (0.05) & 0.329 (0.00) \\
    k-means $L^2$       & 0.013 (0.00) &  0.014 (0.00) & 0.211 (0.08) & 0.202 (0.08) \\
    k-means $W_0^{2,1}$ &-0.060 (0.00) &  0.617 (0.32) & 0.227 (0.06) & {\bf 0.337} (0.08) \\
    \hline
  \end{tabular}
  \caption{
    \label{tab:kmaCompare}
    ARI for clustering 2 and 3 state EEG data using the THMM with
    Wiener measure and functional k-means with $L^2$ and $W_0^{2,1}$
    metrics.  20 random starts were used to get the average performance
    of each method with standard deviation in parentheses.  
  }
\end{table}

\section{Cumulative Snowfall Curves}
\label{sec:snowfall}

An alternative application of the THMM is to model climate data.
In this section, we consider 50 years (winters) of cumulative
snowfall growth curves from the city of Edmonton, Alberta
as recorded by the Meteorological Service of Canada 
(see \url{https://climate.weather.gc.ca/}).
The data was collected at station number 3012208
at latitude 53\degree 34'24 N longitude 113\degree 31'06 W, which 
was the location of the now-closed City Centre Airport.
The data considered spans from the winter of 1940/1941
until the winter of 1989/1990 and was collected 
daily.  A Gaussian kernel smoother
was used to preprocess these growth curves prior to 
analysis.

\subsection{ Two-State Models }

If only the total snowfall is considered as a univariate
time series, a classic Gaussian HMM model can be fit
using the \texttt{HiddenMarkov} package in R 
\citep{HiddenMarkov}.  This naturally splits winters
into high snowfall (mean = 192.9 cm) and low snowfall
(mean = 111.9) years.  Similarly, a two-state THMM
model under the $W_0^{2,1}$ norm also splits the 
winters into heavier and lighter snowfalls.  However,
the heavy snowfall category contains those years with
consistently higher snowfall over the entire winter;
i.e. the heavy snowfall curves are shifted up and left.

Figure~\ref{fig:SnowTwoState} displays the results of 
both the HMM and the THMM fit to total and cumulative 
snowfall, respectively.  In both cases, the fitted
transition matrix indicates that back-to-back heavy 
snowfall winters are unlikely:
$$
  \tilde{A}_\text{HMM} = \begin{pmatrix}
    0 & 1 \\
    0.25 & 0.75
  \end{pmatrix},
  ~~~
  \tilde{A}_\text{THMM} = \begin{pmatrix}
    0 & 1 \\
    0.17 & 0.83
  \end{pmatrix}.
$$
However, the THMM focuses on consistently heavy snowfall
throughout the winter whereas the HMM only is concerned
with the final total.

\begin{figure}
    \centering
    \includegraphics[width=0.475\textwidth]{\PICDIR/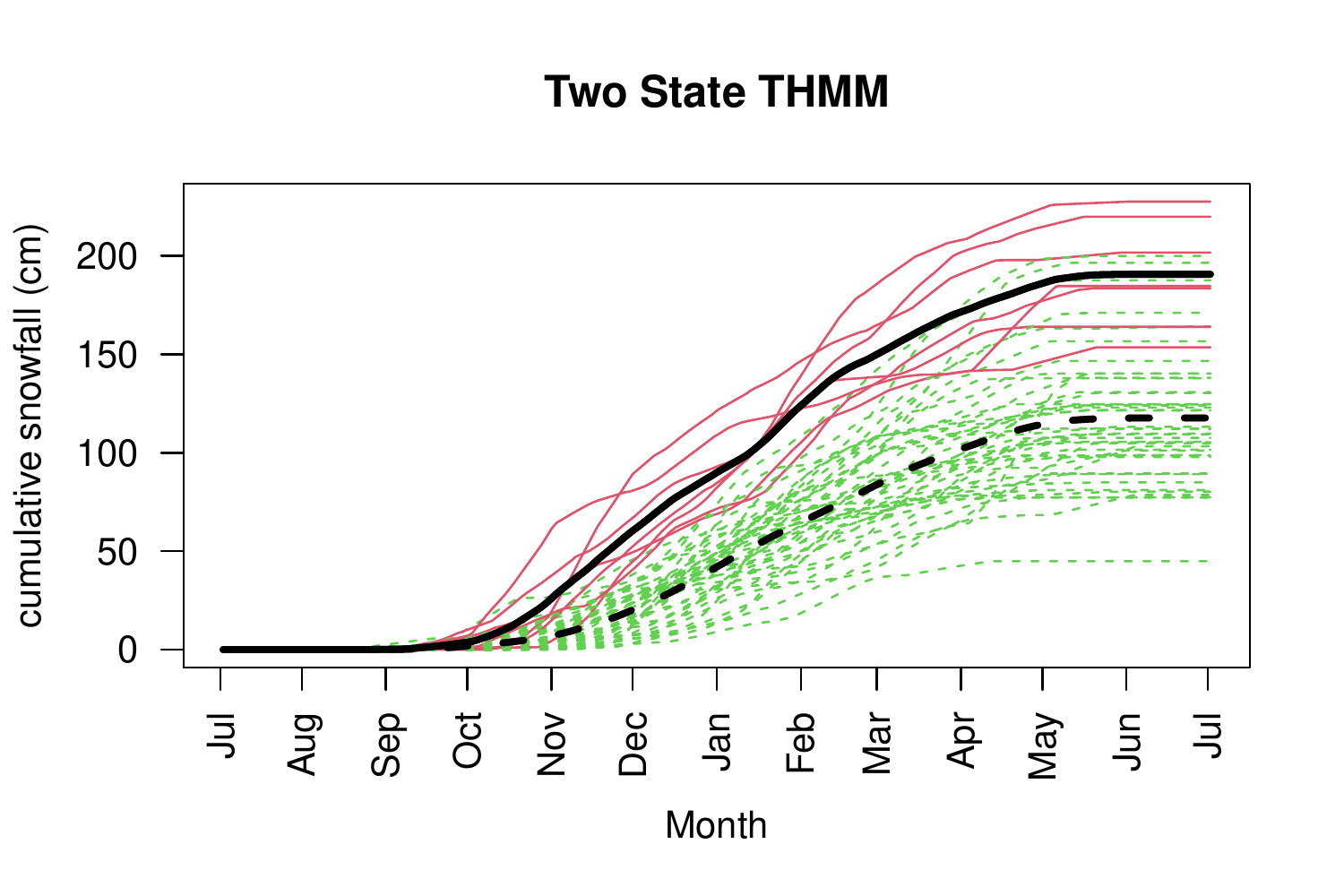}
    \includegraphics[width=0.475\textwidth]{\PICDIR/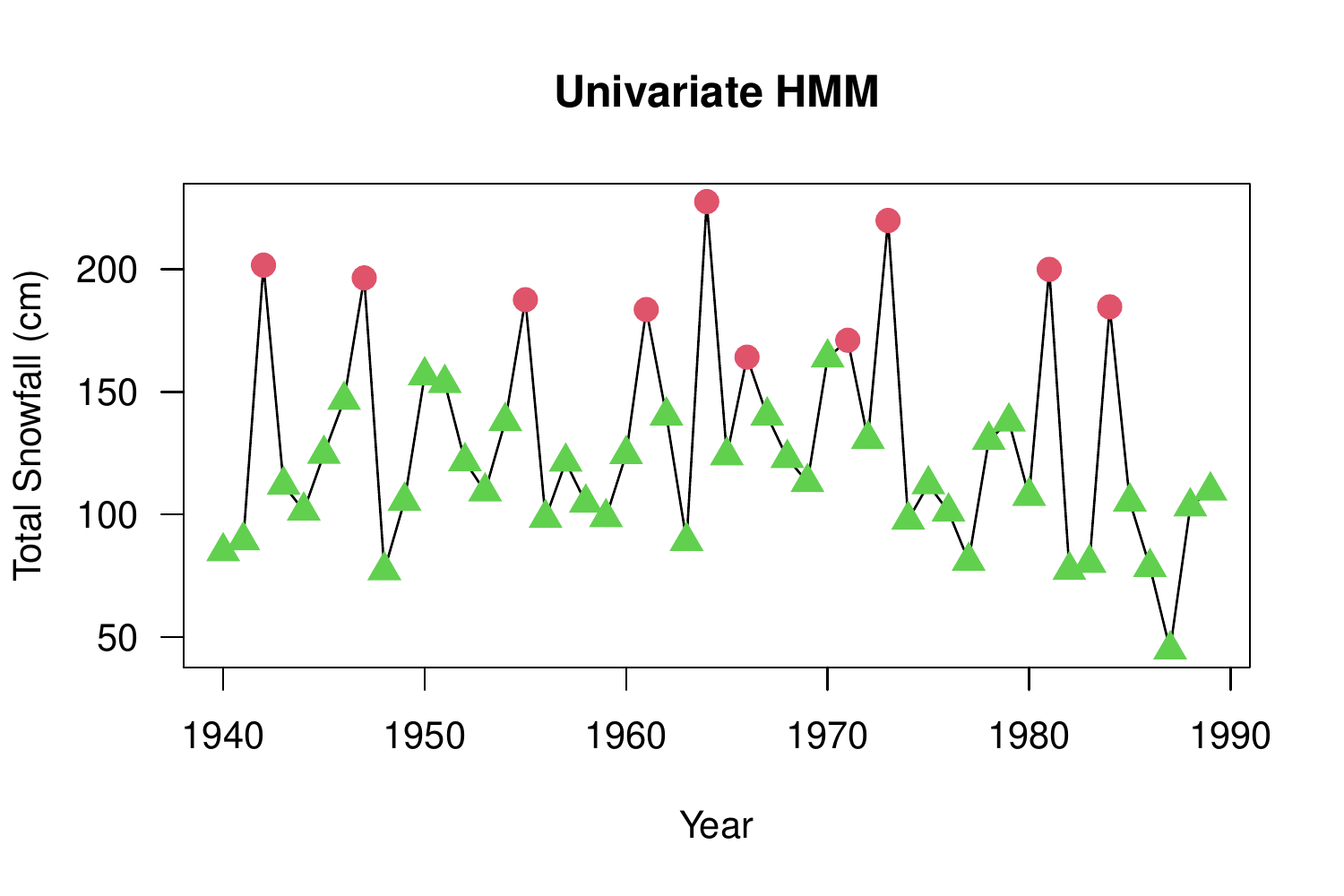}
    \caption{
      \label{fig:SnowTwoState}
      Fitted two state THMM (left) and HMM (right) for
      cumulative and total snowfall in the city of Edmonton
      Alberta, respectively.
    }
\end{figure}

\subsection{ Four-State THMM }

A four-state THMM model applied to these 
smoothed cumulative snowfall growth curves results in 
more interesting findings.  Mainly, it identifies 
three Markov states corresponding to low, medium, 
and high snowfall.  The fourth state is reserved
for the winter of 1954/1955 alone, which had little snowfall
until 18 April when 47.5 cm of snow fell during a three-day
storm.  This event was the largest recorded snowfall in 
Edmonton's recorded history.  The fitted 
transition matrix is
$$
  \tilde{A} = \begin{pmatrix}
     0 & 1 & 0 & 0 \\
     0 & 0 & 0.43 & 0.57 \\
     0 & 0.40 & 0.52 & 0.08 \\
     0.05 & 0.15 & 0.18 & 0.62
  \end{pmatrix}.
$$
Note that $\tilde{A}_{2,2}=0$ indicating 
that back-to-back heavy snowfall years do not occur 
in this 50 year dataset.  This is consistent with 
the two-state models discussed above.
A plot of the data with the four mean curves superimposed
on top is displayed in Figure~\ref{fig:SnowFourState}.

\begin{figure}
    \centering
    \includegraphics[width=0.7\textwidth]{\PICDIR/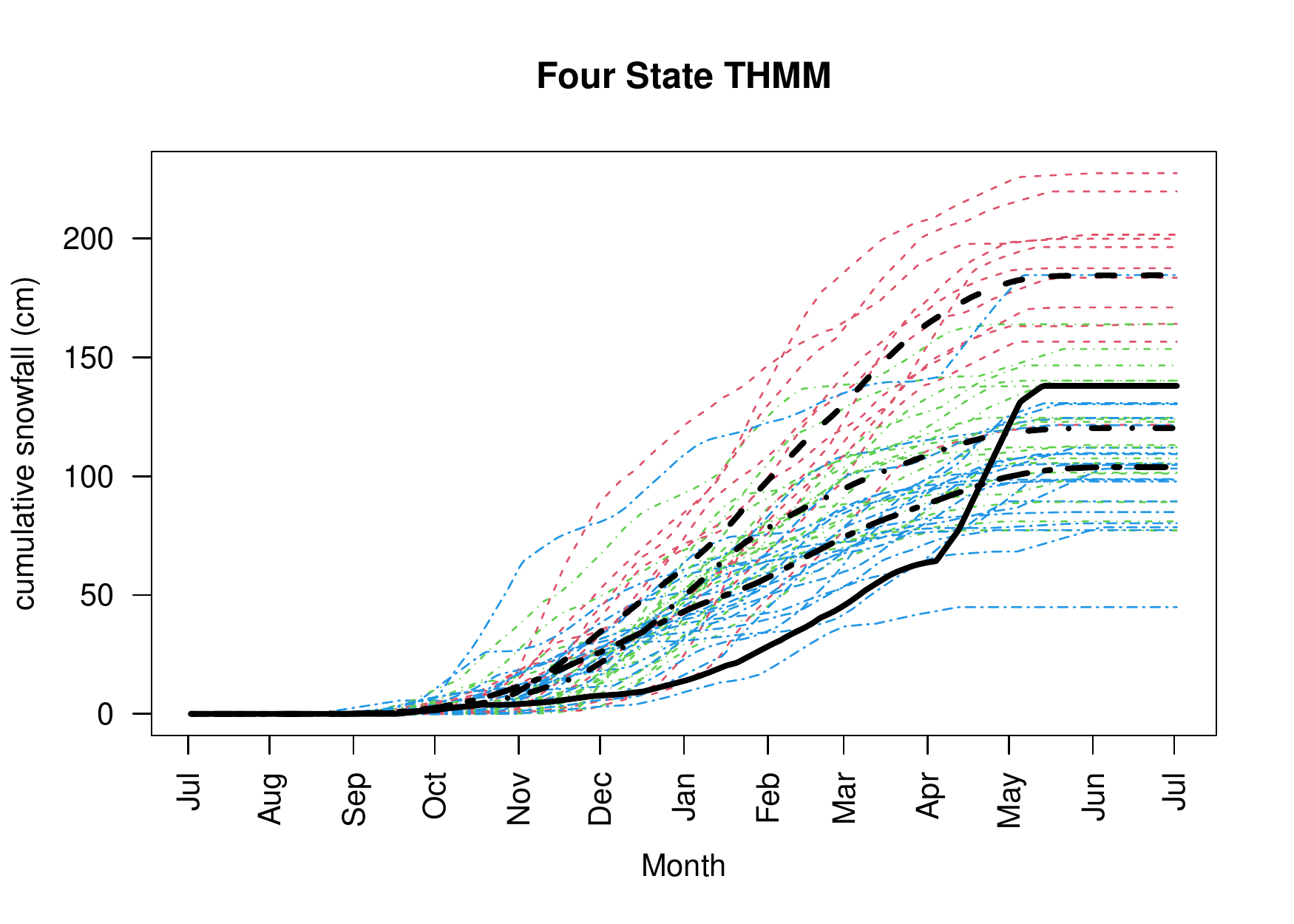}
    \caption{
      \label{fig:SnowFourState}
      Fitted four state THMM for
      cumulative and total snowfall in the city of Edmonton
      Alberta, respectively.
    }
\end{figure}

\section{Future Investigations}
\label{sec:discussion}

\subsection{Estimation of the Hurst Parameter}
\label{sec:hurstEstimation}

In this work, we only consider fractional Brownian motion
with a fixed Hurst parameter chosen by the analyst.  
However much past research 
has gone into estimation of the Hurst parameter from 
data; see \cite{BERZIN2008,KUBILIUS2012,HU2019} and others. 
Future work
could incorporate estimation of the Hurst parameter within
the THMM algorithms.  This would allow for models to be fit
where each Markov state will have its own estimated Hurst 
parameter and thus its own Gaussian Measure / Cameron-Martin 
space.  

\subsection{Regularized Learning}
\label{sec:regularizedLearning}

The Onsager-Machlup functional for a Gaussian measure 
as first presented in \cite{BOGACHEV1998} is 
$$
  \lim_{\veps\rightarrow0}
  \frac{\gamma(V_\veps+h)}{\gamma(V_\veps+k)}
  = \exp\left(
    \frac{1}{2}\abs{\pi_qk}_H^2 - \frac{1}{2}\abs{\pi_qh}_H^2
  \right).
$$
In this work, we set $k=0$ and $h = O_t - h_j$.  However, 
retaining $k$ in the above equation results in the 
following optimization problem:
$$
  \argmin{h_j\in H}
  \sum_{t=1}^T\alpha_t(j)\beta_t(j)\left\{
    \abs{O_t - h_j}_H^2 - \abs{k}_H^2
  \right\}.
$$
Thus, choosing a nonzero $k$ to be a function of $h_j$
will affect the final fit of the THMM.

\subsection{Extensions Beyond the HMM}
\label{sec:extensions}

The Hidden Markov Model is an eminently useful modelling tool.  
However, there are many models that extend and complicate the
beautiful simplicity of the original HMM.  Future work can 
consider extending our proposed THMM from a finite number 
of discrete states to a countably infinite state space
using ideas from the infinite HMM \citep{BEAL2001}.
Continuous state space HMMs have also been considered
with respect to the Kalman filter.
Other ways of adding dependency exist as well including 
autoregressive HMMs that are discussed in 
\cite{JUANG1985} and 
\cite{RABINER1989} as well as in the recent 
works of \cite{LAWLER2019} and \cite{SIDROW2021} 
who use this tool to model animal behaviour. 
Lastly, a Topological Hidden Markov Random Field
could be implemented to model spatial time series
of climate data.

\bibliographystyle{plainnat}
\def\BIBDIR{latexFiles}
\bibliography{./kasharticle,./kashbook,./kashpack,./kashself,./prachi,./giseon}

\end{document}